\documentclass[onecolumn,showpacs,preprint,amsmath,amssymb,superscriptaddress]{revtex4-1}

\usepackage{dcolumn}% Align table columns on decimal point
\usepackage{bm}% bold math
\newcommand{\ignore}[1]{}
\usepackage{graphicx}
\usepackage{subfigure}
\usepackage{float}
\usepackage{amsmath,amssymb,amsfonts,amsthm,color,latexsym}
\usepackage{soul}
\usepackage{cancel}
\usepackage[export]{adjustbox}
\usepackage{tabstackengine}

\def\bm{\boldsymbol}
\newcommand{\Rmnum}[1]{\uppercase\expandafter{\romannumeral #1\relax}}
\newcommand{\rmnum}[1]{\lowercase\expandafter{\romannumeral #1\relax}}

\newcommand*\diff{\mathop{}\!\mathrm{d}}

\mathchardef\mhyphen="2D
\usepackage{amsmath}
\usepackage{xcolor}

\newtheorem{theorem}{Theorem}[section]

\newtheorem{proposition}[theorem]{Proposition}

\usepackage[normalem]{ulem}
\usepackage[hidelinks]{hyperref}

\usepackage{booktabs}
\usepackage{threeparttable}
\usepackage{makecell}

\theoremstyle{definition}

\theoremstyle{remark}

\allowdisplaybreaks[4]
\begin{document}
\preprint{}

\title{Data-driven construction of a generalized kinetic collision operator from molecular dynamics}
\author{Yue Zhao}
\affiliation{Department of Computational Mathematics, Science \& Engineering, Michigan State University, MI 48824, USA}%
\author{Joshua W. Burby}
\affiliation{Department of Physics, the University of Texas at Austin, TX 78712}%
\author{Andrew Christlieb}
\affiliation{Department of Computational Mathematics, Science \& Engineering, Michigan State University, MI 48824, USA}%
\affiliation{Department of Mathematics, Michigan State University, MI 48824, USA}%
\author{Huan Lei}
\email{leihuan@msu.edu}
\affiliation{Department of Computational Mathematics, Science \& Engineering, Michigan State University, MI 48824, USA}%
\affiliation{Department of Statistics \& Probability, Michigan State University, MI 48824, USA}%
%\affiliation{Michigan State University, MI 48824.}%


\begin{abstract}
We introduce a data-driven approach to learn a generalized kinetic collision operator directly from molecular dynamics. Unlike the conventional (e.g., Landau) models, the present operator takes an anisotropic form that accounts for a second energy transfer arising from the collective interactions between the pair of collision particles and the environment. Numerical results show that preserving the broadly overlooked anisotropic nature of the collision energy transfer is crucial for predicting the plasma kinetics with non-negligible correlations, where the Landau model shows limitations.
\end{abstract}

\maketitle
\section{Introduction}
Kinetic theory provides a useful framework for modelling the non-equilibrium processes of a collection of particles such as gas and plasma. A closed-form binary operator is generally introduced to model the collective effects of the micro-scale particle interactions.  One common choice is the Landau collision operator \cite{landau1937kinetic, killeen1976methods, hinton1983collisional}, which assumes the interactions are chaotic and dominated by small-angle scattering. While the operator can accurately characterize the collisions that occur frequently with small perturbations to particle velocities (e.g., high-temperature, weakly coupled plasma), the operator shows limitations for scenarios with significant short-range interactions or large angle scattering. The Boltzmann collision operator \cite{boltzmann1872weitere, wild1951boltzmann, arkeryd1972boltzmann} explicitly accounts for the deflection angle via the cross-section of the collision kernel. The model assumes independent prior-collision velocities and relies on an empirical form of the cross-section, and hence may not be sufficient for plasma with non-negligible correlations (e.g., lower temperature). In particular, the plasma could undergo a broad range of physical conditions in applications related to inertial confinement fusion \cite{Rinderknecht_kinetic_plasma_reivew_2018, Indirect_drive_ICF_PRL_1_2024}, photon scattering \cite{Dharma_Perrot_PRE_1998}, stopping power \cite{Grabowski_Surh_PRL_2013} and ultracold plasmas \cite{Killian_Science_2007}; the applicability of the collision operators remains less understood.

In principle, the BBGKY hierarchy \cite{kirkwood1946statistical, kirkwood1947statistical, bogoliubov1947kinetic, bhatnagar1954model} enables us to systematically encode the particle interactions into the kinetic equation. However, the derived collision term generally relies on the un-closed higher-order correlations. Delicate closed forms such as the Balescu-Lenard operator \cite{lenard1960bogoliubov, balescu1960irreversible} explicitly model the collective many-particle shielding effects beyond the binary interactions, the formulation involves formidable computational complexity which poses limitations to the practical applications. While semi-analytical studies \cite{Baalrud_Daligault_PRL_2013, Daligault_Baalrud_PRL_2016} incorporate the effective potential of mean force into the cross-section of the Boltzmann operator and achieve improved prediction of the transport coefficients; the quantitative modeling of the full kinetic equation that retains effective correlations remains an open problem.     

This work presents a data-driven approach by learning a collision operator directly from the microscale molecular dynamics (MD) \cite{ciccotti1987simulation, frenkel2023understanding}. The operator takes a generalized metriplectic form \cite{morrison1986paradigm, Morrison_dissipative_Phys_Lett_A_1984} that strictly preserves the mass, momentum, energy conservation, and non-negative entropy production. Unlike most existing collision operators, the present model faithfully captures the inhomogeneous scattering in the plane perpendicular to the particle's relative velocities. This inhomogeneity essentially accounts for the collective interactions between the pair of collision particles and the background particles that have been broadly overlooked in most existing empirical forms but prove crucial when the particle correlations become non-negligible (e.g., lower temperature). Moreover, the present model involves a lower 3-dimensional (3D) computational complexity than that of the 5D Boltzmann and the 7D Balescu-Lenard operator. 

To numerically construct the operator from the full MD is non-trivial, as direct learning involves estimating a 3D probability density function (PDF) of the particle velocities, which could become computationally intractable. To overcome the numerical difficulty, we reformulate the learning process in a weak form so that direct PDF estimation can be transformed into Monte Carlo sampling \cite{hastings1970monte, gilks1995markov, frenkel2023understanding} over the pair of particle velocities. Furthermore, we circumvent the expensive double summation based on the random-batch sampling algorithm \cite{jin2020random, li2020random, jin2021random} proposed for electrostatic computation to achieve efficient training with $\mathcal{O}(10^6)$ particles. We examine the constructed operator with the one-component plasma (OCP). % and the one governed by the Yukawa potential \cite{yukawa1935interaction, kremer1986phase}. 
Numerical results show that the present operator can accurately model the kinetic processes in a much broader Coulomb coupling regime, where the canonical Landau form shows limitations. The improvement reveals the crucial role of the broadly overlooked heterogeneous nature of the collisional interactions %arising from the coupling with background particles 
which, fortunately, can be captured by the present model.  
%show the crucial role of the heterogeneous nature of the collision scattering arising from the coupling with background interactions. While all the operators how  

%%%%=============
\section{Methods}
%% ================
%\textbf{2. Model derivation.}
Let us consider the kinetic process of a spatially homogeneous plasma system where the mean field term does not have a net contribution. Let $f(\bm{v},t)$ denote the PDF of the particle velocity $\bm{v} \in \mathbb{R}^{3}$ at time $t \in \mathbb{R}^{+}$.  Without loss of generality, we further assume that the total momentum is zero, i.e., $\rho \bar{\bm v} =  \int \bm v f(t,\bm v){\rm d} \bm v \equiv 0$. The time evolution is governed by the collision operator $\mathcal{C}[f]$ taking the bilinear form \cite{morrison1986paradigm}, i.e., 
\begin{equation}\label{eq:collision1}
    \dfrac{\partial f(\bm{v},t)}{\partial t} = \nabla \cdot \int \bm{\omega} \left[ f(\bm{v}') \nabla f(\bm{v}) - f(\bm{v}) \nabla' f(\bm{v}') \right] \mathrm{d}\bm{v}',
\end{equation}
where $\bm{\omega}(\bm v, \bm v') \in \mathbb{R}^{3\times 3}$ is a  kernel representing the particle collisional interactions. In particular, by choosing $\bm{\omega}(\bm v, \bm v') \propto \vert\bm{u}\vert^{-1}\bm{\mathcal{P}}$, $\bm u = \bm v- \bm v'$  and $\bm{\mathcal{P}} = (\bm{I} - \bm{u}\bm{u}^T/\vert \bm{u}\vert ^{2})$, the collision operator recovers the canonical Landau form which essentially assumes that the interactions are isotropic in the plane perpendicular to $\bm u$ and therefore could show limitations for plasma with non-negligible correlation.  Alternatively, the Balescu-Lenard operator accounts for the anisotropic shielding effect by explicitly modeling a frequency-dependent dielectric function, which, however, leads to a 7D integro-differential equation with formidable computational complexity. 

%To overcome the above limitations, the main issue lies in faithfully modeling the heterogeneous collisional interactions in the $\bm v$-space. This motivates us to seek a generalized representation of  

To overcome the above limitations, we seek a generalized representation of the collision kernel $\bm \omega(\bm v, \bm v')$ that can faithfully capture the heterogeneous collisional interactions in the velocity space and, meanwhile, strictly preserve the physical and symmetry constraints. In particular, we show that if $\bm \omega(\bm v, \bm v')$ is symmetric non-negative definite and satisfies
\begin{equation}\label{eq:app_conditions}
\begin{aligned}
    \bm{\omega}(\mathcal{U}\bm{v},\mathcal{U}\bm{v}') &= \mathcal{U}\bm{\omega}(\bm{v},\bm{v}')\mathcal{U}^T \\
    \bm{\omega}(\bm{v},\bm{v}') &= \bm{\omega}(\bm{v}',\bm{v}) \\
    \bm{\omega}(\bm{v},\bm{v}')(\bm{v} -\bm{v}') &= \bm{0} ,
\end{aligned}
\end{equation}
where $\mathcal{U}$ is a unitary matrix, the collisional operator in Eq. \eqref{eq:collision1} strictly conserves the mass, momentum and energy, and preserves the frame indifference constraints. Furthermore, it ensures non-negative entropy production and admits the Maxwellian distribution as the equilibrium state. We refer to Appendix \ref{sec:ConstructCol} for the detailed proof. 

The above observation motivates us to propose a generalized collision kernel taking the form 
\begin{equation}\label{eq:NN3}
    \bm{\omega} = \bm{\mathcal{P}} \left(g_{r}^2 \widetilde{\bm{r}}\widetilde{\bm{r}}^T + g_{s}^2 \widetilde{\bm{s}}\widetilde{\bm{s}}^T\right) \bm{\mathcal{P}},
\end{equation}
where $\bm{r} = \bm{v} + \bm{v}'$  and $\bm{s} = \bm{u} \times \bm{r}$.
%(or equivalently $\bm{r} = \bm{v} + \bm{v}'$ and  $\bm{s} = \bm{v} \times \bm{v}'$)
We denote $\widetilde{\bm{u}} = \bm{u}/|\bm{u}|$, $\widetilde{\bm{r}} = \bm{\mathcal{P}}\bm{r}/|\bm{\mathcal{P}}\bm{r}|$ and $\widetilde{\bm{s}} = \bm{s}/|\bm{s}|$ as orthonormal vectors.
% represent two directions orthogonal to $\bm u$; $\widetilde{\bm{r}} = \bm{r}/|\bm{r}|$ and $\widetilde{\bm{s}} = \bm{s}/|\bm{s}|$ are unit vectors.
$g_{r} = g_{r}(u,s,t)$ is a rotational invariant scalar function that depends on the magnitude of particle velocities where $u = |\bm{u}|$, $r=|\bm{r}|$ and $s = |\bm{s}|$, and a similar form is applied to $g_{s}(\cdot)$. For systems with non-zero mean velocity, we define $\bm r = \bm v + \bm v' - 2\bar{\bm v}$.
We can show that model \eqref{eq:NN3} strictly satisfies Eq. \eqref{eq:app_conditions}; see Appendix \ref{sec:ConstructCol} for details. 

Compared with the Landau form $\bm{\omega} \propto u^{-1}\bm{\mathcal{P}}$, the new model \eqref{eq:NN3} retains a heterogeneous representation of the local energy transfer process in the plane orthogonal to $\bm u$. Specifically, the term $\widetilde{\bm{r}} \widetilde{\bm{r}}^T$ further relies on the average velocity of the pair of collision particles $\left(\bm{v} + \bm{v}'\right)/2$ and $g_{r} \neq g_{s}$ in general. 
Physically, this enables us to account for the broadly overlooked collective interactions between the pair of particles and the surrounding environment that generally lead to inhomogeneous collisional interactions for plasma with a non-negligible correlation; we postpone the detailed discussion to Fig. \ref{fig:fig1}.   
%where the assumption of independent binary collisions generally shows limitation.
%
This inhomogeneous nature is somewhat similar to the state-dependent effect of the memory term \cite{Lei2010, hijon2010mori, Lyu_Lei_PRL_2023, Ge_Lei_GLE_PRL_2024} that governs the non-equilibrium dynamics of stochastic reduced models. One important difference is that the memory function characterizes the heterogeneous energy dissipation process in the coordinate space, while the present collision kernel characterizes the anisotropic energy transfer process that conserves the total energy in the velocity space.  As shown later, faithfully capturing this inhomogeneous nature is crucial for accurate prediction of the plasma kinetics in the weak coupling regime. In the remainder of this work, we denote the present collision model \eqref{eq:NN3} as CM2 where ``2'' refers to the two subspaces spanned by $\widetilde{\bm{r}}$ and $\widetilde{\bm{s}}$. As a special case, we also consider the collision model with a homogeneous kernel $\bm\omega = g_{I}^{2}(u) \bm{\mathcal{P}}$ which will be denoted by CM1 and can be viewed as a generalization of the Landau model (i.e.,  by choosing $g_I^2(u) \propto 1/u$). 

% In particular, this which generally leads to inhomogeneous collisional interactions

% For plasma with non-negligible correlation,  the assumption of independent binary collisions generally shows limitation; the present form enables us to account for the broadly overlooked collective interactions between the pair of particles and the surrounding environment.  

To learn the collision model \eqref{eq:NN3}, we construct $g_{r}(\cdot)$ and $g_{s}(\cdot)$ by matching the time evolution of the particle velocity PDF $f(\bm v, t)$ predicted from the kinetic model \eqref{eq:collision1} and the full micro-scale MD simulation. While we use the neural network in this study, other forms such as the kernel representation can be also used. To train the model, one essential challenge is the accurate estimation of the 3D velocity PDF $f_{\rm MD}(\bm v, t)$ from the empirical distribution $f^{\ast} (\bm v) = 1/N\sum_{i=1}^N \delta(\bm v- \bm v_i)$, where $\bm v_i$ is the velocity of the individual MD particles.  To circumvent this difficulty, we construct the empirical loss by transferring  the direct evaluation of the density evolution into a weak form, i.e., 
\begin{equation}
L=  \sum_{m=1}^{N_m} \sum_{l=1}^{N_t} \sum_{k=1}^{N_\psi}  \left(\left.\dfrac{\partial f^{(l,m)}}{\partial t}\right|_{\rm MD}-\left.\dfrac{\partial f^{(l,m)}}{\partial t}\right|_{\rm c}, \psi_k\right)^2, 
\label{eq:loss_weak_form}
\end{equation}
where $(f, \psi) := \int f(\bm v)\psi(\bm v)\diff \bm v \approx \frac{1}{N} \sum_{i=1}^N \psi(\bm v_i)$ represents an inner product evaluated by the empirical distribution of the MD particles and $\left\{\psi_k(\bm v)\right\}_{k=1}^{N_\psi}$ is a set of test bases. 
$f^{(l,m)} = f^{(m)}(\bm v, t_l)$ represents the density at $t=t_{l}$ and $m$ is the index representing the MD simulations starting with various initial distributions. In this work,  we choose $N_m = 3$ and collect the training samples by setting the initial PDF  $f(\bm v, 0)$ following the uniform, bi-Maxwellian, and double-well distributions; see Appendix \ref{sec:TrainingDetails} for the detailed form of $f(\bm v, 0)$ and $\psi(\bm v)$. We emphasize that the constructed model will be validated by the kinetic processes different from the ones used for training. 

The MD density evolution in the weak form \eqref{eq:loss_weak_form} can be evaluated by 
\begin{equation}
\left(\left.\dfrac{\partial f^{(l)}}{\partial t}\right|_{\rm MD}, \psi_k\right) = \dfrac{1}{N\delta t}\sum_{i=1}^N (\psi_k(\bm{v}_{i}^{l+1})-\psi_k(\bm{v}_{i}^{l})),
\label{eq:loss_MD_weak_form}
\end{equation}
where index $m$ is skipped for simplicity and $\delta t$ is the time step of the training sample. The summation can be efficiently pre-computed over the MD particles.  On the other hand, the weak form of the kinetic density evolution, with the integral by parts, takes the form
\begin{equation}
\begin{split}
&\left(\left.\dfrac{\partial f}{\partial t}\right|_{c},\psi_k\right) =
        \left(f(\bm{v})\int \bm{\omega}_{ij} f(\bm{v}')\mathrm{d}\bm{v}', \dfrac{\partial \psi_k(\bm{v})}{\partial v_i \partial v_j}\right) \\
        &\quad+ \left( f(\bm{v}) \int \left[\dfrac{\partial \bm{\omega}_{ij}}{\partial \bm{v}_j} - \dfrac{\partial \bm{\omega}_{ij}}{\partial \bm{v}_j '}\right] f(\bm{v}') \mathrm{d}\bm{v}', \dfrac{\partial \psi_k (\bm{v})}{\partial \bm{v}_i} \right),
\end{split}    
\end{equation}
where indices $i$ and $j$ follow the Einstein summation. In particular, we note that the integration with respect to $\bm v$ and $\bm v'$ leads to a double summation over the MD particles with $\mathcal{O}(N^2)$ complexity, which becomes computationally intractable for the common MD simulation (e.g., $N\sim \mathcal{O}({10}^6)$). To overcome the limitation, we use the random mini-batch approach \cite{ketkar2017stochastic, jin2021random} as an unbiased stochastic approximation, i.e.,
\begin{equation}
    \begin{split}
        &\left(\left.\dfrac{\partial f}{\partial t}\right|_{c},\psi_k\right) \approx \dfrac{1}{P}\sum_{p=1}^{P} \bm{\omega}_{ij}(\bm{v}_{n(p)},\bm{v}_{n'(p)} ') \dfrac{\partial^2 \psi_k (\bm{v}_{n(p)})}{\partial \bm{v}_i \partial \bm{v}_j} \\
        &~+ \left[ \dfrac{\partial \bm{\omega}_{ij}}{\partial \bm{v}_j} - \dfrac{\partial \bm{\omega}_{ij}}{\partial \bm{v}_j '} \right](\bm{v}_{n(p)},\bm{v}_{n'(p)} ') \dfrac{\partial\psi_k (\bm{v}_{n(p)})}{\partial \bm{v}_i}, 
  \end{split}
  \label{eq:loss_kinetic_weak_form}
\end{equation}
where $\left\{n(p), n'(p)\right\}_{p=1}^P$ represents a set of pairs of indices randomly chosen from the full MD samples with $1\le n(p), n'(p)\le N$ and $n(p)\neq n'(p)$ for each training step. In this work, we choose $P=10^5$, which enables us to establish efficient training of generalized collision kernel $\bm\omega$ by minimizing empirical loss function in form of Eqs. \eqref{eq:loss_weak_form} \eqref{eq:loss_MD_weak_form} \eqref{eq:loss_kinetic_weak_form}; see Appendix \ref{sec:TrainingDetails} for training details.

\section{Numerical results}
The present collision operator enables us to investigate the kinetic processes of plasma systems in a broader physical regime compared with the conventional model.
In this work, we study a monovalent cation OCP system consisting of $10^6$ charged particles in a $10^2\times10^2\times10^2$\AA$^3$ domain with a periodic boundary condition imposed on each direction.
The MD simulations are conducted under equilibrium configurations with several initial velocity distributions (see Appendix \ref{sec:MDSet}). The constructed collision operators will be examined by kinetic processes different from the training ones. 
For validation purposes, we first consider the plasma under high thermal energy $k_BT=100$ eV, where the corresponding plasma coupling parameter $\Gamma = \frac{q_e^2}{4\pi \epsilon_0 k_B T}\left(4\pi n/3\right)^{1/3}$ is $0.23$ and the assumption of small-angle scattering remains valid. As expected, the predictions from the kinetic equation with both the present and the Landau operator show good agreement with the full MD simulation results 
(see Appendix \ref{sec:HighTemp}). In the remainder of this work, we focus on the more challenging lower thermal energy regime with $k_BT= 10$ eV, where the Landau collision operator generally shows limitations for $\Gamma \sim \mathcal{O}(1)$.

\begin{figure}[h]
    \centering
    \includegraphics[width=0.8\textwidth]{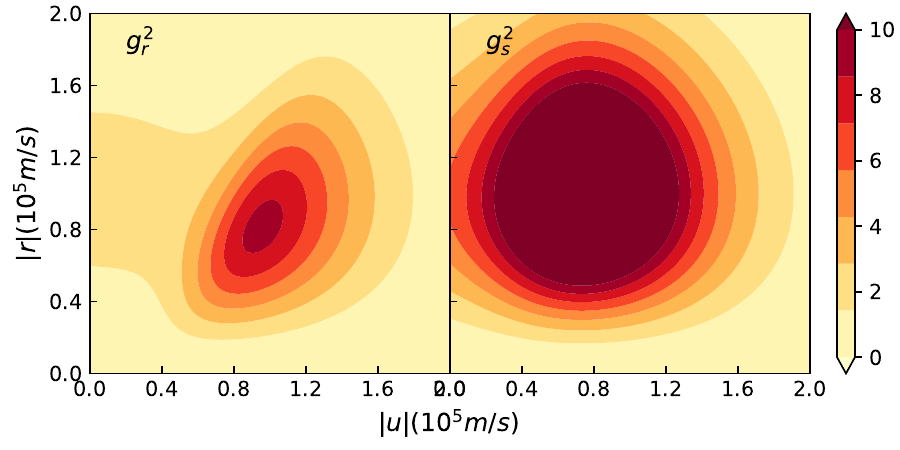}
    \caption{The energy transfer magnitude functions $\mathbb{E}_{s}[g_{r}^{2}(u,r,s)]$ and $\mathbb{E}_{s}[g_{s}^{2}(u,r,s)]$ of the present collision model CM2 in Eq. \eqref{eq:NN3}. The ensemble average is taken over the component $s$ under the equilibrium distribution. Unlike the Landau model, $g_{r}^2 < g_{s}^2$ implies the anisotropic nature of the energy transfer arising from the collective interaction between the pair of particles and the environments.}
    \label{fig:fig1}
\end{figure}

Let us start with the generalized collision kernel $\bm\omega(\bm v, \bm v')$ in the form of Eq. \eqref{eq:NN3}. Fig. \ref{fig:fig1} shows the projection of the magnitude functions $\mathbb{E}_{s}[g_{r}^2(u, r, s)]$ and $\mathbb{E}_{s}[g_{s}^2(u, r, s)]$ constructed from the MD samples, where the ensemble average is taken over the component $s$ under the equilibrium distribution. In contrast to the isotropic form (i.e., $g_{r}=g_{s}$) of the Landau operator, the magnitude of the energy transfer $g_{r}^2$ along the direction $\widetilde{\bm{r}} \propto \bm{\mathcal{P}}\bm{r}$ is generally smaller than $g_{s}^2$ along the direction $\widetilde{\bm{s}} \propto \bm{u} \times \bm{r}$ (see Appendix \ref{sec:LowTemp} for additional results). This anisotropic nature can be understood as below. In the high-temperature chaotic regime (i.e., $\Gamma \ll 1$), the energy transfer is dominated by the binary collision of two particles with relative velocity $\bm u = \bm v- \bm v'$ and is isotropic in the plane orthogonal to $\bm u$, i.e., $g_{r}^2 = g_{s}^2 \equiv g_{u,\perp}^2$.
As temperature decreases (i.e., $\Gamma \sim \mathcal{O}(1)$), the particle correlation becomes non-negligible and the collective interactions between the pair of particles and the environment begin to play a role. Accordingly, there exists a second energy transfer from the collective motion $\bm r \propto (\bm v + \bm v')/2$ to the orthogonal plane.
However, due to the energy conservation, this energy transfer is restricted to the null space of $\bm u$ (i.e., the projection by $\bm{\mathcal{P}} = \bm{I} - \bm{u}\bm{u}^T/\vert \bm{u}\vert ^{2}$).
As a result, this collective interaction will lead to a net energy transfer from $\widetilde{\bm{r}} \propto \bm{\mathcal{P}}\bm{r}$ to $\widetilde{\bm{s}} \propto \bm{u} \times \bm{r}$, i.e., $g_{r}^2 = g_{u,\perp}^2 - \delta g_{r, \perp}^2$ and $g_{s}^2 = g_{u,\perp}^2 + \delta g_{r, \perp}^2$.
While the theoretical derivation of the analytical form for this second energy transfer $\delta g_{r, \perp}^2$ is highly non-trivial, the present data-driven approach provides a faithful way to accurately capture this broadly overlooked effect.

\begin{figure}[h]
    \centering
    \includegraphics[width=0.8\textwidth]{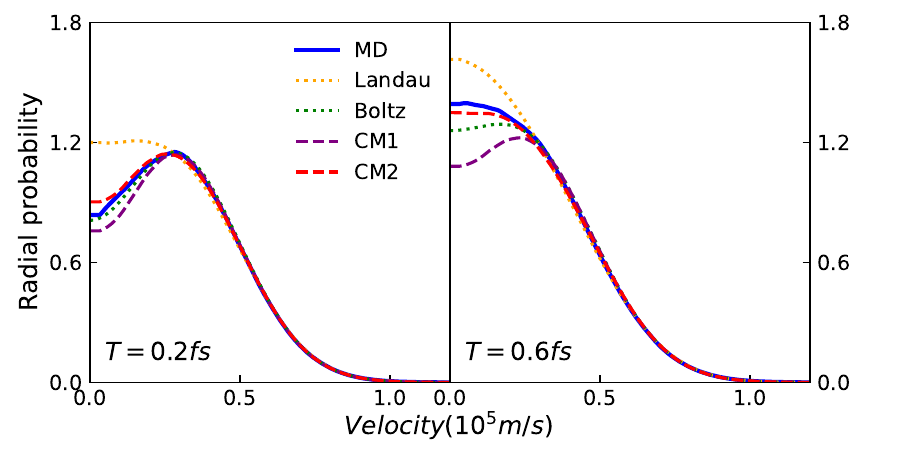}
    \caption{The instantaneous distribution of the radial velocity magnitude with the BKW model as the initial condition predicted by the full MD simulations and the kinetic equation with the Landau, Boltzmann, and the present CM1 and CM2 collision models at (a) $t=0.2$ fs and (b) $t=0.6$ fs.}
    \label{fig:fig2}
\end{figure}

To probe the anisotropic effects on plasma kinetics, we study the dynamical processes with various initial distributions. We emphasize that these processes are not included in the training set of the present collision model. First, we examine the relaxation from the initial velocity PDF $f(\bm v, t=0)$ following the Bobylev-Krook-Wu (BKW) model \cite{krook1977exact, bobylev1975exact}.  
%where the exact solution of the Boltzmann equation is known analytically for $t>0$. 
We note that the exact BKW solution of the Boltzmann equation corresponds to the Maxwellian molecule rather than the Coulomb interactions and we should not expect the solution to agree with the MD results. It is merely used as a benchmark problem to generate the non-equilibrium initial condition for comparative study. 
Fig. \ref{fig:fig2} shows the predictions of the distribution of the radial velocity magnitude. For a short time at $t=0.2$ fs, the prediction of the kinetic equation with both the present (CM1 and CM2) %and the Boltzmann collision operator 
shows good agreement with the full MD results. The larger discrepancy of the Landau model indicates the limitation of its energy transfer magnitude function ($\propto u^{-1}$). Furthermore, for a longer time $t= 0.6$ fs, the predictions from the kinetic equation with both the empirical and the present CM1 collision model show apparent deviations from the MD results. In contrast, the present CM2 model yields good agreement, implying the crucial role of the anisotropic nature of the energy transfer process.

\begin{figure}[h]
    \centering
    \includegraphics[width=0.8\textwidth]{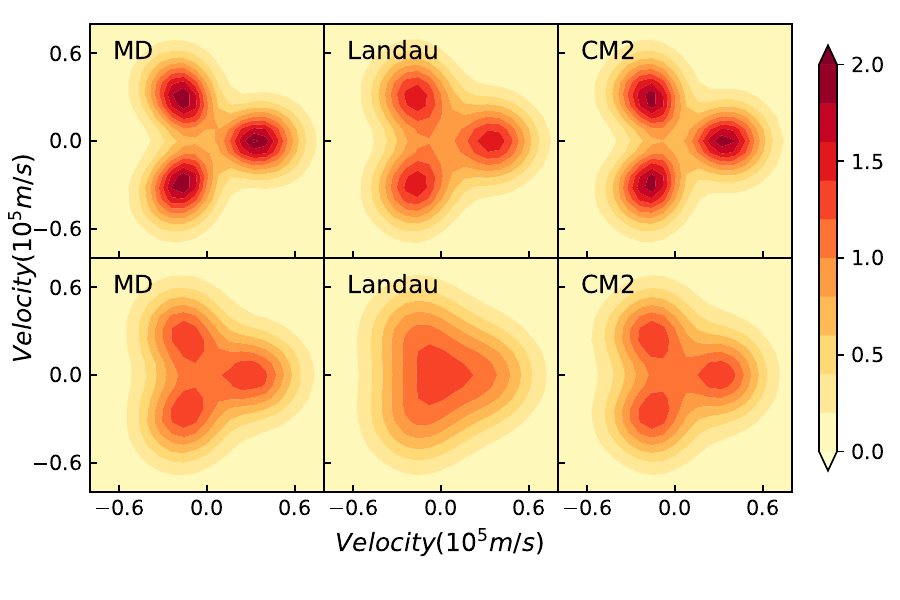}
    \caption{The instantaneous velocity PDF in the $v_{1}\mhyphen v_{2}$ plane from a trimodal initial distribution predicted by the full MD, the Landau and the present CM2 collision model at $t = 0.2~\text{fs}$ (upper) and $0.6~\text{fs}$ (lower).}
    \label{fig:fig3}
\end{figure}

Next, we investigate the kinetic process with the initial velocity PDF taking a trimodal distribution, where the particle velocity is
equally concentrated at three positions in the $v_1\mhyphen v_2$ plane, and is Maxwellian in the $v_3$ direction. As shown in Fig. \ref{fig:fig3},  the prediction of the Landau model overestimates the relaxation process towards the equilibrium distribution and shows apparent deviations from the MD results.  This overestimation is likely due to the aforementioned anisotropic energy transfer effect. In particular, the collective motion is dominated in the $v_1\mhyphen v_2$ plane that results in larger energy transfer in the $v_3$ direction. On the other hand, the Landau model assumes the isotropic form and therefore leads to overestimation of the relaxation process. Furthermore, we emphasize that this discrepancy can not be simply remedied by scaling the Landau model, which, however, leads to larger discrepancies for other kinetic processes (see SM \ref{sec:LowTemp} for additional results). Fortunately, this complex anisotropic effect can be faithfully encoded in the present collision model, which yields good agreement with the full MD results.

Finally, we examine the kinetic process from a 3D double-well initial velocity distribution which is asymmetric in the $v_1\mhyphen v_2$ plane. Similar to the previous trimodal distribution, the Landau model overestimates the relaxation process, as shown in Fig. \ref{fig:fig4}. On the other hand, the present CM2 collision model accurately reproduces the full MD results. 

\begin{figure}[h]
    \centering
    \includegraphics[width=0.8\textwidth]{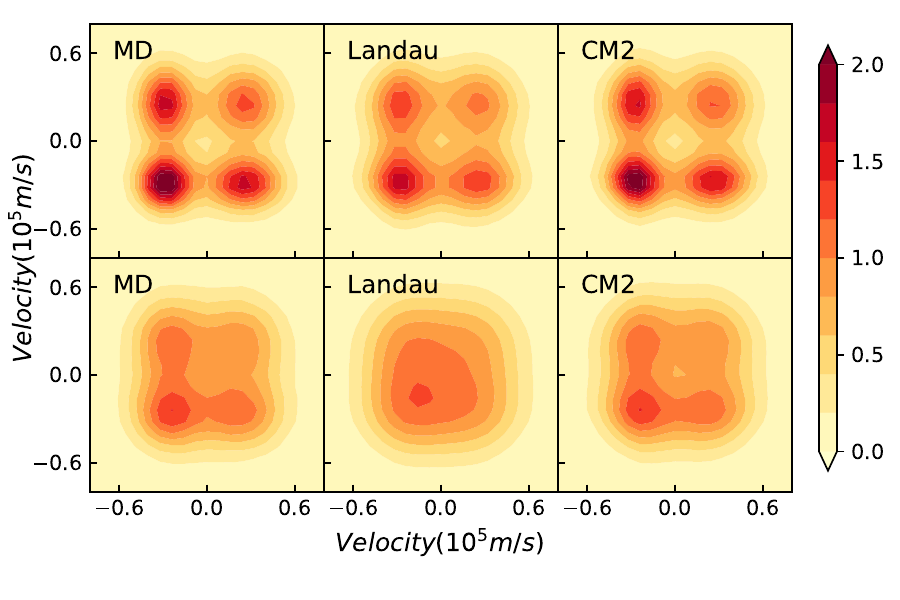}
    \caption{The instantaneous velocity PDF in the $v_{1}\mhyphen v_{2}$ plane from an asymmetric double-well distribution predicted by the full MD, the Landau and the present CM2 collision model at $t = 0.2~\text{fs}$ (upper) and $0.6~\text{fs}$ (lower).}
    \label{fig:fig4}
\end{figure}

\section{Summary}
In summary, this work presents a data-driven approach to learning a generalized collision operator that strictly preserves the physics and symmetry constraints directly from the micro-scale MD models. The constructed model reveals the anisotropic nature of the collisional energy transfer arising from the collective interactions between the pair of particles and the environment. This anisotropic effect has been broadly overlooked in common empirical collision models such as the Landau form but proves to be crucial for plasma with non-negligible correlations. In particular, the generalized collisional model significantly broadens the applicability of the collisional kinetic description for the weak coupling regime ($\Gamma \sim \mathcal{O}(1)$) where the empirical forms show limitations. 
Essential ideas of existing efficient numerical methods \cite{Carrillo_JCP_2020,caflisch2008hybrid,caflisch2016accelerated,taitano2021conservative} for the Landau collisional kinetic model can be naturally applied to the present model. 
Furthermore, the present random-batch-based weak formulation essentially provides a general framework for efficient learning of spatially inhomogeneous and multi-species meso-scale kinetic models from micro-scale MD descriptions governed by various particle interactions, 
%and can be potentially used to learn other collision operators such as the Boltzmann equation.   
%
and paves the way towards plasma kinetics 
%for applications with spatial inhomogeneity, multi-species, and stronger correlation ($\Gamma \sim \mathcal{O}(10)$). We leave this for future study. 
%for applications where the effects of spatial inhomogeneity, multi-species, and stronger correlation ($\Gamma \sim \mathcal{O}(10)$) need to be considered. We leave this for future study. 
%and paves the way towards the plasma kinetics 
for %spatially inhomogeneous and multi-species 
systems with stronger correlations. 
%($\Gamma \sim \mathcal{O}(10)$) where the effects of the many-body and non-Markovian interactions need to be properly introduced. 
We leave this for future study.

\begin{acknowledgements}
We acknowledge helpful discussions from Philip J. Morrison. The work is supported in part by the National Science Foundation under Grant DMS-2110981, the Department of Energy under Grant No. DOE-DESC0023164 and the ACCESS program through allocation MTH210005.
\end{acknowledgements}

\appendix

\section{Proposition and construction of the generalized collision operator}\label{sec:ConstructCol}

In this work, we start with the kinetic dynamics of spatially homogeneous one-component plasma (OCP) in the metriplectic form \cite{morrison1986paradigm}, where 
%the potential energy and kinetic energy of the system are weakly coupled or decoupled.
the evolution of the velocity probability density function (PDF) $f(\bm v,t)$ satisfies
\begin{equation}\label{eq:collision}
    \dfrac{\partial f(\bm{v},t)}{\partial t} = \nabla \cdot \int \bm{\omega} \left[ f(\bm{v}') \nabla f(\bm{v}) - f(\bm{v}) \nabla' f(\bm{v}') \right] \mathrm{d}\bm{v}',
\end{equation}
%We assume the OCP is uniformly distributed in phase space, and the velocities are randomly selected from an initial distribution with zero total momentum.
where $\bm{\omega}(\bm v, \bm v') \in \mathbb{R}^{3\times 3}$ is a symmetric positive semi-definite kernel representing the particle collisional interactions. It is designed to satisfy the following properties: invariance under variable exchange, symmetry under rotation, and zero projection onto $\bm v-\bm v'$, i.e., 
\begin{equation}\label{eq:conditions}
\begin{aligned}
    \bm{\omega}(\mathcal{U}\bm{v},\mathcal{U}\bm{v}') &= \mathcal{U}\bm{\omega}(\bm{v},\bm{v}')\mathcal{U}^T \\
    \bm{\omega}(\bm{v},\bm{v}') &= \bm{\omega}(\bm{v}',\bm{v}) \\
    \bm{\omega}(\bm{v},\bm{v}')(\bm{v} -\bm{v}') &= \bm{0},
\end{aligned}
\end{equation}
where $\bm{\mathcal{U}}$ is a unitary matrix.

\begin{proposition}\label{prop:property}
With the collision operator satisfying Eq. \eqref{eq:conditions},
the kinetic model \eqref{eq:collision},  strictly conserves the mass, momentum, and energy, and preserves the frame indifference constraints.
Furthermore, it ensures non-negative solution and entropy production,
and admits the Maxwellian distribution as the equilibrium state.
\end{proposition}

\begin{proof}
Without loss of generality, we assume the particle mass is unit. The total density, momentum and kinetic energy are defined as
    \begin{equation}\label{eq:PhysProp}
        n = \int f \mathrm{d} \bm{v}, ~ \bm{p} = \int \bm{v}f \mathrm{d} \bm{v}, ~ K = \frac{1}{2}\int \bm{v}^2 f \mathrm{d} \bm{v}.
    \end{equation}
    For any physical quantity $\phi$, we have
    \begin{equation}\label{eq:ConservationLaw}
        \begin{aligned}
            \dfrac{\mathrm{d}}{\mathrm{d} t} \int \phi(\bm{v}) f \mathrm{d} \bm{v} &= \int \phi(\bm{v}) \left(\nabla \cdot \int \bm{\omega} \left[ f(\bm{v}') \nabla f(\bm{v}) - f(\bm{v}) \nabla' f(\bm{v}') \right] \mathrm{d}\bm{v}'\right) \mathrm{d}\bm{v}\\
            &= -\dfrac{1}{2}\iint \left(\nabla_{v} \phi(\bm{v}) - \nabla_{v'} \phi(\bm{v}')\right) \bm{\omega}(\bm{v},\bm{v}') \left[ f(\bm{v}') \nabla f(\bm{v}) - f(\bm{v}) \nabla' f(\bm{v}') \right] \mathrm{d}\bm{v}' \mathrm{d}\bm{v}.
        \end{aligned}
    \end{equation}
    It is easy to show the quantities $(n,\bm{p},K)$ are conserved.
    The entropy of the system and its evolution are given by 
    \begin{equation}\label{eq:Entropy}
        \begin{aligned}
            S(f)&=-\int f \log f \mathrm{d}\bm{v},\\
            \dfrac{\mathrm{d}S}{\mathrm{d}t} = \dfrac{1}{2}\iint & B \bm{\omega}(\bm{v},\bm{v}') B f(\bm{v})f(\bm{v}') \mathrm{d}\bm{v}' \mathrm{d}\bm{v} \geq 0,
        \end{aligned}
    \end{equation}
where $B = \nabla_{v} \log f(\bm{v}) - \nabla_{v'} \log f(\bm{v}')$ and we have used the fact that $\bm\omega \succeq 0$.  

Furthermore, the evolution of the PDF is positive-preserving based on the proposed collision operator. At the point where the PDF first reaches zero, i.e., $f(\bm{v}_{0},t)=0$, the gradient of $f$ also vanishes, such that $\nabla f(\bm{v}_{0},t) = \bm{0}$. The time evolution equation satisfies
    \begin{equation}
        \frac{\partial f(\bm{v}_{0},t)}{ \partial t}  = \nabla^2 f : \left(\int \bm{\omega}f'\mathrm{d}\bm{v}'\right) \geq 0,
    \end{equation}
since $\nabla^2 f$ is a non-negative definite matrix and $\int \bm{\omega}f'\mathrm{d}\bm{v}'$ is a symmetric positive semi-definite matrix.
    
Finally, we show that the equilibrium distribution corresponds to the Maxwellian distribution. By applying the Lagrangian multiplier method, we have 
\begin{equation}
        L = \int f\ln f \mathrm{d}\bm{v} + \lambda_{1} \left(\int f \mathrm{d}\bm{v} - 1\right) + \lambda_{2} \left(\int \bm{v}f \mathrm{d}\bm{v}\right) + \lambda_{3} \left(\int \dfrac{1}{2}\bm{v}^{2} f \mathrm{d}\bm{v} - \dfrac{3}{2} k_{B} T\right),
\end{equation}
where $\lambda_{1}$, $\lambda_{2}$ and $\lambda_{3}$ are Lagrangian multipliers. When $L$ reaches its minimum value (maximum entropy), we have 
    \begin{equation}
        \dfrac{\partial L}{\partial \lambda_{1}}=0,~~~
        \dfrac{\partial L}{\partial \lambda_{2}}=0,~~~
        \dfrac{\partial L}{\partial \lambda_{2}}=0,~~~
        \dfrac{\delta L}{\delta f}=0,
    \end{equation}
where the equilibrium state admits the Maxwellian distribution $f_{eq}\propto \exp(-\vert \bm{v}\vert^{2}/2k_{B}T)$.
\end{proof}

Furthermore, we show that the generalized collision kernel needs to satisfy the following constraint.
\begin{proposition}\label{prop:symmetry}
The collision kernel needs to satisfy the symmetry condition $\bm{\omega}(\bm{v},\bm{v}') = \bm{\omega}(-\bm{v},-\bm{v}')$.
\end{proposition}
\begin{proof}
Let us consider $h(\bm{v},t_{0})=f(-\bm{v},t_{0})$ at the initial time $t_{0}$, then $h(\bm{v},t)=f(-\bm{v},t)$ should be satisfied at $t\geq t_{0}$.
The evolution of $h(\bm{v},t)$ should satisfy the above general collision equation
    \begin{equation*}
        \dfrac{\partial h}{\partial t}(\bm{v},t) = \nabla \cdot \int \bm{\omega}(\bm{v},\bm{v}') \left[ h(\bm{v}') \nabla h(\bm{v}) - h(\bm{v}) \nabla' h(\bm{v}') \right] \mathrm{d}\bm{v}' .
    \end{equation*}
On the other hand
    \begin{equation*}
        \begin{aligned}
            \dfrac{\partial h}{\partial t}(\bm{v},t) &= \dfrac{\partial f}{\partial t}(-\bm{v},t) \\
            &= -\nabla \cdot \int \bm{\omega}(-\bm{v},\bm{v}') \left[ -f(\bm{v}') \nabla f(-\bm{v}) - f(-\bm{v}) \nabla' f(\bm{v}') \right] \mathrm{d}\bm{v}' \\
            &= -\nabla \cdot \int \bm{\omega}(-\bm{v},-\bm{v}') \left[ -f(-\bm{v}') \nabla f(-\bm{v}) + f(-\bm{v}) \nabla' f(-\bm{v}') \right] \mathrm{d}\bm{v}' \\
            &= \nabla \cdot \int \bm{\omega}(-\bm{v},-\bm{v}') \left[ h(\bm{v}') \nabla h(\bm{v}) - h(\bm{v}) \nabla' h(\bm{v}') \right] \mathrm{d}\bm{v}',
        \end{aligned}
    \end{equation*}
where we insert $-\bm{v}$ into the collision function Eq. \eqref{eq:collision} in the first step, change variable $\bm{v}' \rightarrow -\bm{v}'$ in the second step, and use $h(\bm{v})=f(-\bm{v})$ in the last step.
\end{proof}

\iffalse
\begin{equation}\label{eq:full}
    \begin{aligned}
        \bm{\omega} &= \bm{\mathcal{P}} \bm{A} \bm{A}^T \bm{\mathcal{P}}, \\
        \bm{A} &= \bm{A}_1 + (\bm{u}^T \bm{r})\bm{A}_2, \\
        \bm{A}_1 &= g_{ss} \bm{r}\bm{r}^T + g_{tu} \bm{s}\bm{u}^T + g_{tt} \bm{s}\bm{s}^T , \\
        \bm{A}_2 &= g_{su} \bm{r}\bm{u}^T + g_{st} \bm{r}\bm{s}^T + g_{ts} \bm{s}\bm{r}^T .
    \end{aligned}
\end{equation}
In the formula, $\bm{A}_1$ and $(\bm{u}^T \bm{r})\bm{A}_2$ remain unchanged when $\{\bm{v},\bm{v}'\} \rightarrow \{\bm{v}',\bm{v}\}$, and the terms related to $\bm{u}$ vanish due to $\bm{\mathcal{P}}\bm{u}=\bm{0}$.
\fi

% where $\bm{\omega} = (|\bm{u}|^2 \bm{I} - \bm{u}\bm{u}^T)/|\bm{u}|^3$ in the Landau equation

To construct the generalized collision operator,
we denote $\bm{u}=\bm{v}-\bm{v}'$, $\bm{r}=\bm{v}+\bm{v}' - 2\bar{\bm v}$, $\bm{s}=\bm{u}\times\bm{r}$, and $\bm{\mathcal{P}}=\bm{I}-\bm{u}\bm{u}^T/|\bm{u}|^2$ as a projection operator, where $\bar{\bm v} = \rho^{-1}\int \bm v f(\bm v, t){\rm d}\bm v$ is the mean velocity. 
Furthermore, we denote $\widetilde{\bm{u}}=\bm{u}/|\bm{u}|$, $\widetilde{\bm{r}}=\bm{\mathcal{P}}\bm{r}/|\bm{\mathcal{P}}\bm{r}|$ and $\widetilde{\bm{s}}=\bm{s}/|\bm{s}|$ as mutually orthogonal unit vectors, satisfying $\widetilde{\bm{u}} \widetilde{\bm{u}}^{T} + \widetilde{\bm{r}} \widetilde{\bm{r}}^{T} + \widetilde{\bm{s}} \widetilde{\bm{s}}^{T} = \bm{I}$.
Accordingly, we can construct the collision kernel as:
\begin{equation}\label{eq:CM3}
    \begin{aligned}
        \bm{\omega} &= \bm{\omega}_{1} + \bm{\omega}_{2} + \bm{\omega}_{3}, \\
        \bm{\omega}_{1} &= g_{1}^{2} \bm{\mathcal{P}} \bm{r} \bm{r}^{T} \bm{\mathcal{P}}, \\
        \bm{\omega}_{2} &= g_{2}^{2} \bm{\mathcal{P}} \bm{s} \bm{s}^{T} \bm{\mathcal{P}}, \\
        \bm{\omega}_{3} &= (\bm{u}^{T}\bm{r}) g_{3} \bm{\mathcal{P}} (\bm{r}\bm{s}^{T} + \bm{s}\bm{r}^{T}) \bm{\mathcal{P}}.
    \end{aligned}
\end{equation}

It satisfies the conditions in Eq. \eqref{eq:conditions}.
However, the cross term $\bm{\omega}_{3}$ doesn't satisfy the symmetry condition $\bm{\omega}(\bm{v},\bm{v}') = \bm{\omega}(-\bm{v},-\bm{v}')$ in Prop. \ref{prop:symmetry}, so we have the general collision model as 
\begin{equation}\label{eq:CM2}
    \bm{\omega} = \bm{\mathcal{P}} \left(g_{r}^2 \widetilde{\bm{r}}\widetilde{\bm{r}}^T + g_{s}^2 \widetilde{\bm{s}}\widetilde{\bm{s}}^T\right) \bm{\mathcal{P}},
\end{equation}
named by ``CM2'' model, where ``2'' refers to the two subspaces spanned by $\widetilde{\bm{r}}$ and $\widetilde{\bm{s}}$.
We can further simplify it into ``CM1'' as $\bm{\omega} = g_{I}^{2}(|\bm{u}|) \bm{\mathcal{P}}$, which reduces to the Landau model \cite{landau1937kinetic} when $g_{I}^{2}(|\bm{u}|) \propto 1/|\bm{u}|$.

The generalized collision operator depends on the tensor product of the vector perpendicular to $\bm{u}$.
The functions $g_{\ast}$ are rotational invariant for $u$, $r$ and $s$, and are represented by neural networks $g_{\ast} = g_{\ast}(u,r,s)$, with $u=|\bm{u}|$, $r=|\bm{r}|$ and $s=|\bm{s}|$. Compared with the Landau \cite{landau1937kinetic} and Boltzmann model \cite{boltzmann1872weitere}, the present generalized collision operator enables us to capture the heterogeneous collisional energy transfer in the plane perpendicular to $\bm{v}-\bm{v}'$ arising from the collective interaction between the pair of collision particles and the environment, which leads to inhomogeneous energy transfer magnitude in directions $\widetilde{\bm{r}} \propto \bm{\mathcal{P}}\bm{r}$ and $\widetilde{\bm{s}} \propto \bm{u} \times \bm{r}$. Moreover, the magnitude functions not only depend on $u$ but also $r$ and $s$.

\section{Training details}\label{sec:Training}
\label{sec:TrainingDetails}
 
To directly train the generalized collision operator from the molecular dynamics (MD) simulations \cite{frenkel2023understanding}, we can rewrite the kinetic equation \eqref{eq:collision} as
\begin{equation}\label{eq:PartialInt}
    \begin{aligned}
        \left.\dfrac{\partial f}{\partial t}\right|_{c} =&
        \dfrac{\partial}{\partial\bm{v}_i}\int \bm{\omega}_{ij}\left(
        f(\bm{v}')\dfrac{\partial f(\bm{v})}{\partial \bm{v}_j}
        - f(\bm{v})\dfrac{\partial f(\bm{v}')}{\partial \bm{v}_j'}\right) \mathrm{d}\bm{v}' \\
        =& \dfrac{\partial}{\partial \bm{v}_i} \left( \dfrac{\partial f}{\partial \bm{v}_j} \int \bm{\omega}_{ij} f(\bm{v}') \mathrm{d}\bm{v}' \right)
        - \dfrac{\partial}{\partial \bm{v}_i} \left( f(\bm{v}) \int \bm{\omega}_{ij}\dfrac{\partial f(\bm{v}')}{\partial\bm{v}_j '} \mathrm{d}\bm{v}' \right) \\
        =& \dfrac{\partial^2}{\partial \bm{v}_i \partial \bm{v}_j} \left( f(\bm{v}) \int \bm{\omega}_{ij} f(\bm{v}') \mathrm{d}\bm{v}' \right)
        - \dfrac{\partial}{\partial \bm{v}_i} \left( f(\bm{v}) \int \left[ \dfrac{\partial \bm{\omega}_{ij}}{\partial \bm{v}_j} - \dfrac{\partial \bm{\omega}_{ij}}{\partial \bm{v}_j'} \right] f(\bm{v}') \mathrm{d}\bm{v}' \right),
        %+ \dfrac{\partial}{\partial \bm{v}_i} \left( f(\bm{v}) \int \dfrac{\partial \bm{\omega}_{ij}}{\partial \bm{v}_j'} f(\bm{v}') \mathrm{d}\bm{v}' \right)
    \end{aligned}
\end{equation}
where we compare the prediction of the time evolution of the PDF between the kinetic model and the MD simulation. One numerical challenge is the accurate estimation of 
the velocity PDF $f_{\rm MD}(\bm v,t)$ from the empirical distribution $f^{(l)} = 1/N \sum_{n=1}^{N} \delta(\bm{v}-\bm{v}_{n})$ at $t=t_{l}$. To alleviate this challenge, we reformulate the Eq. \eqref{eq:PartialInt} and construct the empirical loss  in a weak form, i.e., 
\begin{equation}\label{eq:loss}
    L=  \sum_{m=1}^{N_m} \sum_{l=1}^{N_t} \sum_{k=1}^{N_\psi}  \left(\left.\dfrac{\partial f^{(l,m)}}{\partial t}\right|_{\rm MD}-\left.\dfrac{\partial f^{(l,m)}}{\partial t}\right|_{\rm c}, \psi_k\right)^2,
    % L=\sum_k \left(\left.\dfrac{\partial f}{\partial t}\right|_{MD}-\left.\dfrac{\partial f}{\partial t}\right|_{c},\psi_k(\bm{v})\right)^2.
\end{equation}
where $\psi_k(\bm{v})$ are test functions, $(\cdot,\cdot)$ represents the inner product, $f^{(l,m)} = f^{(m)}(\bm v, t_l)$ represents the density at $t=t_{l}$ and $m$ is the index representing the MD simulations starting with various initial distributions. In this work, we choose $N_{m}=3$ and select the uniform, bi-Maxwellian, and symmetric double-well distributions as the initial velocity distribution, see detailed form in Eq. \eqref{eq:trainingset} in Sec. \ref{sec:MDSet}. For the test functions, we choose $N_\psi = 4$ and use $\psi_{k}(\bm{v}) = \exp(-\bm{v}^2)$, $\bm{v}^2 \exp(-\bm{v}^2)$, $\exp[-(\bm{v}^2 - C_0)^2]$, and $\exp[-(\bm{v} - \bm{c}_0)^2]$ where $C_{0}$ and $\bm{c}_{0}$ are constants. 

To compute the empirical loss, we note that the MD part of the inner product in Eq. \eqref{eq:PartialInt} can be precomputed by the empirical distribution, i.e., 
\begin{equation}\label{eq:WeakMD}
    \left(\left.\dfrac{\partial f}{\partial t}\right|_{\rm MD}, \psi_k(\bm{v})\right) =\left(\dfrac{f^{l+1}-f^{l}}{\delta t},\psi_k(\bm{v})\right)
    =\dfrac{1}{N\delta t}\sum_{n=1}^{N} (\psi_k(\bm{v}_{n}^{l+1})-\psi_k(\bm{v}_{n}^{l})).
\end{equation}
where we omit $m$ for simplicity. 
However, the kinetic collision part of the inner product in Eq. \eqref{eq:PartialInt} involves double summations over the MD particles,
\begin{equation}\label{eq:WeakCol}
    \begin{aligned}
        \left(\left.\dfrac{\partial f}{\partial t}\right|_{c},\psi_k(\bm{v})\right) =&
        \left(f(\bm{v})\int \bm{\omega}_{ij} f(\bm{v}')\mathrm{d}\bm{v}', \dfrac{\partial \psi_k(\bm{v})}{\partial \bm{v}_i \partial \bm{v}_j}\right)
        + \left( f(\bm{v}) \int \left[ \dfrac{\partial \bm{\omega}_{ij}}{\partial \bm{v}_j} - \dfrac{\partial \bm{\omega}_{ij}}{\partial \bm{v}_j '} \right] f(\bm{v}') \mathrm{d}\bm{v}', \dfrac{\partial \psi_k (\bm{v})}{\partial \bm{v}_i} \right) \\
        =& \dfrac{1}{N^2}\sum_{n,n'}^{N} \bm{\omega}_{ij}(\bm{v}_{n},\bm{v}_{n'} ') \dfrac{\partial^2 \psi_k}{\partial \bm{v}_i \partial \bm{v}_j}(\bm{v}_{n}) + \left[ \dfrac{\partial \bm{\omega}_{ij}}{\partial \bm{v}_j} - \dfrac{\partial \bm{\omega}_{ij}}{\partial \bm{v}_j '} \right](\bm{v}_{n},\bm{v}_{n'} ') \dfrac{\partial\psi_k}{\partial \bm{v}_i} (\bm{v}_{n}) ,
    \end{aligned}
\end{equation}
which can be considered as the statistical average of $\dfrac{\partial \bm{\omega}_{ij}}{\partial \bm{v}_j}$, $\dfrac{\partial \bm{\omega}_{ij}}{\partial \bm{v}_j'}$,  $\bm{\omega}_{ij}$ and derivatives of $\psi_k(\bm{v})$ over the samples from MD simulations. %at time $t$
However, since $N\sim \mathcal{O}(10^6)$, the direct pairwise summation over $\{\bm{v}_{n},\bm{v}_{n}'\}$ is not feasible. To circumvent this difficulty, we can use the mini-batch approach \cite{jin2020random, jin2021random}, i.e., 
\begin{equation}\label{eq:WeakColRBM}
    \left(\left.\dfrac{\partial f}{\partial t}\right|_{c},\psi_k(\bm{v})\right) \approx \dfrac{1}{P}\sum_{p=1}^{P} \bm{\omega}_{ij}(\bm{v}_{n(p)},\bm{v}_{n'(p)} ') \dfrac{\partial^2 \psi_k}{\partial \bm{v}_i \partial \bm{v}_j}(\bm{v}_{n(p)}) + \left[ \dfrac{\partial \bm{\omega}_{ij}}{\partial \bm{v}_j} - \dfrac{\partial \bm{\omega}_{ij}}{\partial \bm{v}_j '} \right](\bm{v}_{n(p)},\bm{v}_{n'(p)} ') \dfrac{\partial\psi_k}{\partial \bm{v}_i} (\bm{v}_{n(p)}) ,
\end{equation}
where $\left\{n(p), n'(p)\right\}_{p=1}^P$ represents a set of pairs of indices randomly chosen from the full MD samples with $1\le n(p), n'(p)\le N$ and $n(p)\neq n'(p)$ for each training step. In this work, we choose $P = 10^5$ which proves to be efficient to train the collision model.  As a special case, for the CM1 model, we note that $\bm{\omega}_{ij} = g_{I}^{2}(|\bm{u}|)\bm{\mathcal{P}}$, $\frac{\partial \bm{\omega}_{ij}}{\partial \bm{v}_j} = - \frac{\partial \bm{\omega}_{ij}}{\partial \bm{v}_j '}$.
%Since the collision operators are presented as neural networks, we use several velocity trajectory data to compute the loss in Eq. \eqref{eq:loss} and optimize network parameters.

The encoder functions $g_r$ and $g_s$ are parameterized as $6$ layer fully connected neural networks. Each hidden layer consists of $10$ neurons.  The networks are trained by Adam \cite{Kingma_Ba_Adam_2015} for $6\times 10^{5}$ steps. For each step, $10^{5}$ pairs of collision particles will be randomly selected as one training batch.  The initial learning rate is $0.01$ and the decay rate is $0.99$ per $2000$ steps.

\section{The setup of MD simulations}\label{sec:MDSet}

In this study, we simulate the OCP of monovalent cation systems consisting of $10^6$ charged particles in a $10^2\times10^2\times10^2$\AA$^3$ domain with a periodic boundary condition imposed on each direction. The particles interact via the Coulomb potential in a homogeneous neutralized background of electrons. In particular, we consider two temperatures with $k_B T = 100 \text{eV}$ and $10 \text{eV}$, where the corresponding plasma coupling parameter $\Gamma$ is $\mathcal{O}(0.1)$ and $\mathcal{O}(1)$, respectively. 
%
%The initial condition of the MD simulations takes the equilibrium distribution of the particle position. The velocity distribution is given by Eqs. \eqref{eq:trainingset} and \eqref{eq:testset}. 
The time integration is conducted by the velocity-Verlet scheme \cite{frenkel2023understanding} with the time step $dt = 10^{-4}~\text{fs}$.  The long-range force is calculated with the particle-particle-particle-mesh Ewald (PPPM) method \cite{hockney1988computer} with $10^{-4}$ relative accuracy.
% The charge density is $10^{24}~ \text{cm}^{-3}$.

For each simulation,  we conduct an isothermal simulation of the target temperature to obtain the equilibrium configurations.  The particle velocity is randomly sampled from various PDFs in the training set in Eq. \eqref{eq:trainingset} and test set in Eq. \eqref{eq:testset}. Specifically, the training set includes the uniform distribution $f_1(\bm{v})$, bi-Maxwellian distribution $f_2(\bm{v})$ and symmetric double-well distribution $f_3(\bm{v})$. The test set includes the trimodal distribution $f_4(\bm{v})$, $\chi^2$ distribution $f_5(\bm{v})$, and radial oscillation distribution $f_6(\bm{v})$.

\noindent\textbf{Training set:}
\begin{subequations}\label{eq:trainingset}
    \begin{align}
        f_1(\bm{v}) &\sim U[-\sqrt{3}\sigma, \sqrt{3}\sigma]^3,\\
        f_2(\bm{v}) &\sim \prod_{i=1}^{3} \exp(-\bm{v}_i^2 / 2\sigma_i^2), \\
        f_3(\bm{v}) &\sim \prod_{i=1}^{3} [\alpha_{i_1}\exp(-(\bm{v}_i-b_{i_1})^2 / 2\sigma_{i_1}^2) + \alpha_{i_2}\exp(-(\bm{v}_i+b_{i_2})^2 / 2\sigma_{i_2}^2)],
    \end{align}
\end{subequations}
\textbf{Test set:}
\begin{subequations}\label{eq:testset}
    \begin{align}
        f_4(\bm{v}) &\sim \bm{v}_{12}^2 \exp(-\bm{v}_{12}^2) \times g(\theta_{12}) \times f_{eq}(\bm{v}_3),\quad 
        g(\theta_{12}) \sim \sum_{n=1}^{3} N(2n\pi/3,1/16), \\ 
        % \quad \tan(\theta) = \bm{v}_2 / \bm{v}_1, \\
        f_5(\bm{v}) &\sim |\bm{v}|^2 \exp(-\bm{v}^2 / 2\sigma^2), \\
        f_6(\bm{v}) &\sim \exp(-\alpha \bm{v}^2) \times [\cos(\bm{v}^2)]^2 .
    \end{align}
\end{subequations}

\iffalse
\begin{equation}\label{eq:initial}
    \begin{aligned}
        f_1(\bm{v}) &\sim U[-\sqrt{3}\sigma, \sqrt{3}\sigma]^3,\\
        f_2(\bm{v}) &\sim \prod_{i=1}^{3} \exp(-\bm{v}_i^2 / 2\sigma_i^2), \\
        f_3(\bm{v}) &\sim \prod_{i=1}^{3} [\alpha_{i1}\exp(-(\bm{v}_i-b_{i1})^2 / 2\sigma_{i1}^2) + \alpha_{i2}\exp(-(\bm{v}_i+b_{i2})^2 / 2\sigma_{i2}^2)], \\
        f_4(\bm{v}) &\sim \bm{v}_{12}^2 \exp(-\bm{v}_{12}^2) \times g(\theta_{12}) \times f_{eq}(\bm{v}_3),\quad 
        g(\theta_{12}) \sim \sum_{n=1}^{3} N(2n\pi/3,1/16), \\ 
        % \quad \tan(\theta) = \bm{v}_2 / \bm{v}_1, \\
        f_5(\bm{v}) &\sim |\bm{v}|^2 \exp(-\bm{v}^2 / 2\sigma^2), \\
        f_6(\bm{v}) &\sim \exp(-\alpha \bm{v}^2) \times [\cos(\bm{v}^2)]^2 .
    \end{aligned}
\end{equation}
\fi

The system is equilibrated by $0.4~\text{fs}$ and followed by a production phase of $20~\text{fs}$.
%The initial equilibrium period accounts for the dependence on inconsistencies in initial velocity and particle position configurations, which are omitted in both the training and simulation processes.
We emphasize that only velocity samples collected from the training set in Eq. \eqref{eq:trainingset} are utilized to learn the generalized collision operators; see details in Sec. \ref{sec:TrainingDetails}. On the other hand, the constructed generalized collision operator will be examined with the kinetic processes in both the training and test set by comparing the predictions with those obtained from empirical Landau and Boltzmann models as well as the full MD simulations. The (time scaling) parameters of the empirical models are chosen as the optimal values by matching the MD results of the relaxation process for the bi-Maxwellian cases.

% Within each MD simulation, we consider the initial velocity distribution after $0.4~\text{fs}$ MD simulation, to prevent any mismatch between the specified initial velocities and the particle position configurations of the system.
% In the Landau equation, the optimal time step is determined by aligning the temperature evolution curves for the initial bi-Maxwellian distributions with initial temperature ratios of $4:4:1$ and $1:1:4$.

% In addition to the OCP model, the Yukawa interaction model\cite{yukawa1935interaction} with Debye shielding length can also be used to simulate plasma systems.
% The Yukawa potential can be derived from the linearized Poisson-Boltzmann equation, as
% \begin{equation}\label{eq:Yukawa}
%     \phi(\bm{r}) = \dfrac{1}{4\pi\epsilon_0} \dfrac{q^2 e^{-\kappa r}}{r} , \qquad
%     \kappa^2 = \dfrac{n_e q^2}{\epsilon_0 k_B T_e},
% \end{equation}
% where $T_e$ is the temperature of electrons, usually higher than the ion temperature.
% The background electrons are treated as an electronic cloud around charges to screen the Coulomb interaction.
% When $T_e\rightarrow\infty$, the Yukawa model becomes the OCP model.
% From the MD results, it is found that when the screened Coulomb potential with the same electron temperature and ion temperature is used, the system evolution process is close to the Landau equation, which is due to the linear Poisson-Boltzmann equation approximation.

Finally, we briefly discuss the implications of the different temperatures considered in this work, which specifies the ratio of the Coulomb energy to the thermal energy and 
quantifies by the plasma coupling parameter $\Gamma$, i.e., 
\begin{equation}\label{eq:CoulpPara}
    \Gamma = \dfrac{q_e^2}{4\pi \epsilon_0 k_B T}\sqrt[3]{\dfrac{4\pi n}{3}}.
\end{equation}
% \href{https://en.wikipedia.org/wiki/Plasma_parameter#Coupling_parameter}{``plasma coupling parameter''}
% and another plasma parameter is the ratio of the maximum impact parameter to the classical distance of closest approach in Coulomb Debye shielding $\Lambda$,
% \begin{equation}\label{eq:CoulpPara2}
%         \Lambda = 4\pi n \lambda_D^3 , \qquad
%         \lambda_D^3 = \left( \dfrac{\epsilon k_B T}{\sum_{j=1}^{N} n_j q_j^2} \right)^{1/2} ,
% \end{equation}
% where $n$ and $n_j$ are the particle number density.
For physical conditions of the present OCP system, $\Gamma$ is $0.232$ and $2.321$ for $k_{B}T = 100 ~ \text{eV}$ and $10 ~ \text{eV}$, respectively. In particular, for the regime $\Gamma \sim \mathcal{O}(0.1)$, the particle correlation is weak; both the present generalized form and the empirical Landau model yield accurate predictions. On the other hand, for the regime $\Gamma \sim O(1)$, %($\Lambda > 1$)
the particle correlation becomes non-negligible which leads to inhomogeneous collisional energy transfer.  The empirical Landau form generally shows limitations. Conversely, the present generalized collision operator can faithfully capture this effect and therefore accurately predict the kinetic processes in this regime; see the following two sections for details.

\section{Validation of the high temperature regime}\label{sec:HighTemp}

To validate the constructed model, let us start with the OCP system at a high temperature of $100 ~ \text{eV}$ by using the ``CM1'' model and the Landau equation, and compare the results with MD simulations.
In all cases, the results obtained from the Landau equation and the ``CM1'' model closely align with the exact MD solutions, validating the accuracy of the proposed collision operator.
The consistency suggests that small-angle scattering is the dominant interaction mechanism at high temperatures, allowing the Landau equation to provide a reliable description of the evolution of the velocity distribution of the particles. Such consistency is also observed in plasma with lower density or higher temperature.

\begin{figure}[H]
    \centering
    \includegraphics[width=0.45\textwidth]{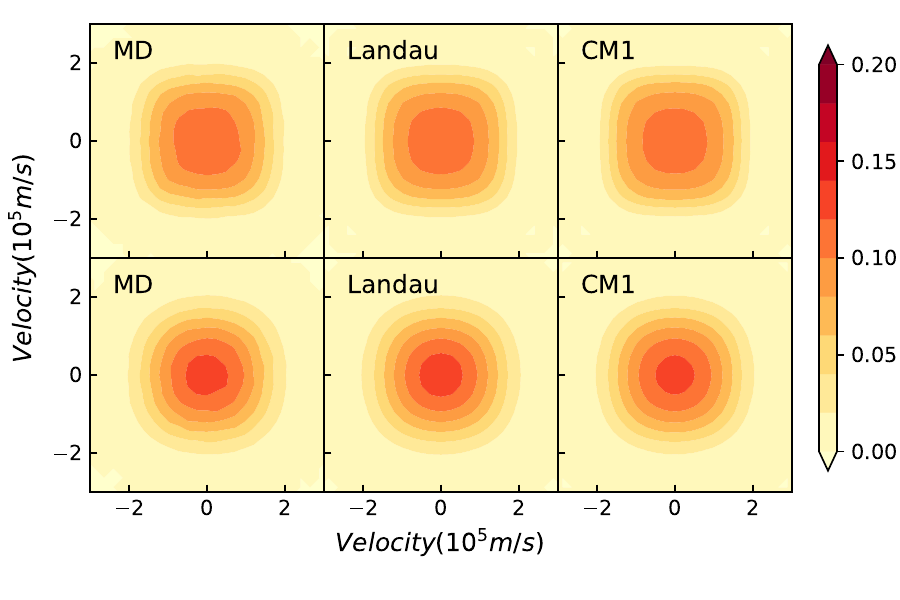}
    \caption{The instantaneous velocity PDF in the $v_{1}\mhyphen v_{2}$ plane from a uniform initial distribution predicted by the full MD, the Landau and the CM1 collision model at $t = 2~\text{fs}$ (upper) and $4~\text{fs}$ (lower).}
    \label{fig:100Uni}
\end{figure}

\begin{figure}[H]
    \centering
    \includegraphics[width=0.45\textwidth]{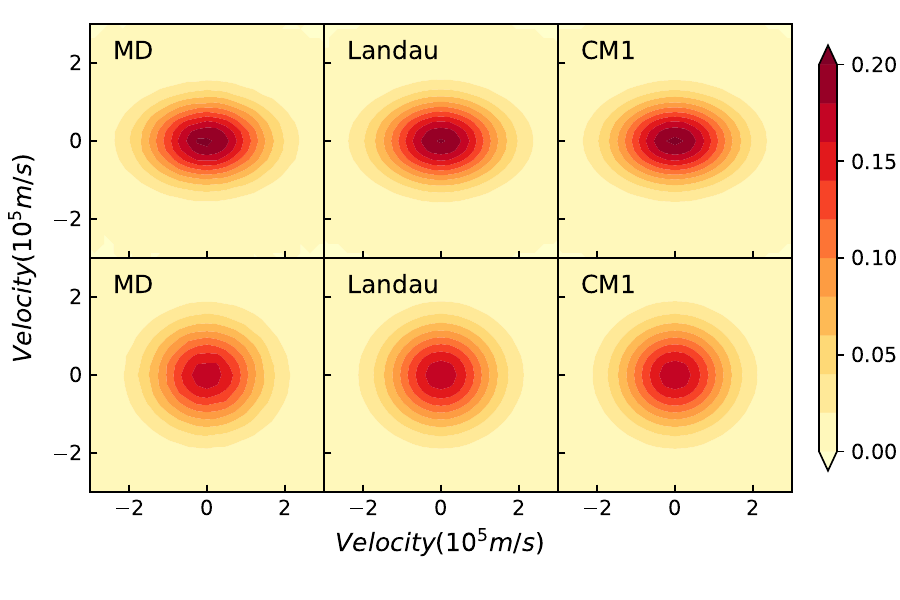}
    \includegraphics[width=0.45\textwidth]{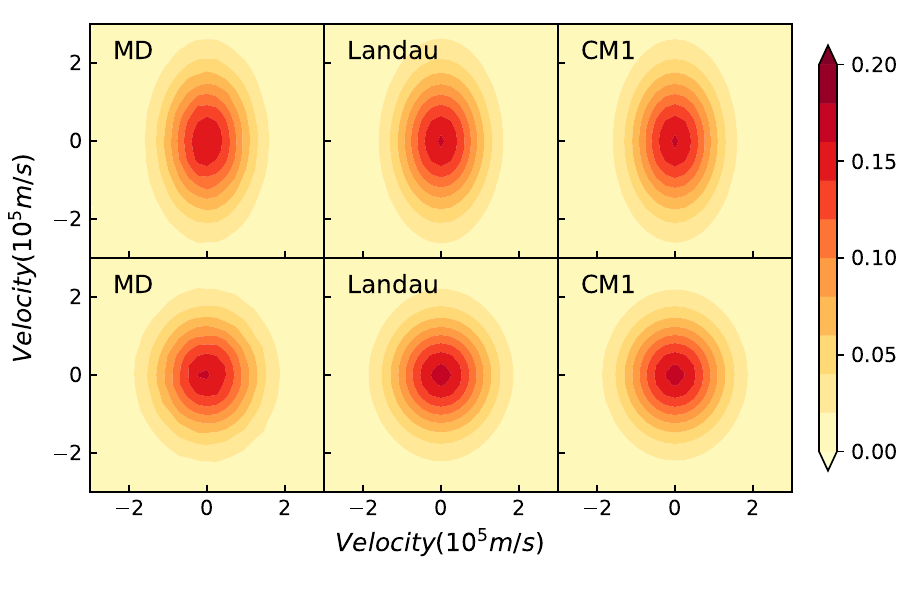}
    \caption{The instantaneous velocity PDF in the $v_{1}\mhyphen v_{2}$ plane from bi-Maxwellian initial distributions predicted by the full MD, the Landau and the CM1 collision model at $t = 2~\text{fs}$ (upper) and $10~\text{fs}$ (lower).}
    \label{fig:100biMax}
\end{figure}

\begin{figure}[H]
    \centering
    \includegraphics[width=0.45\textwidth]{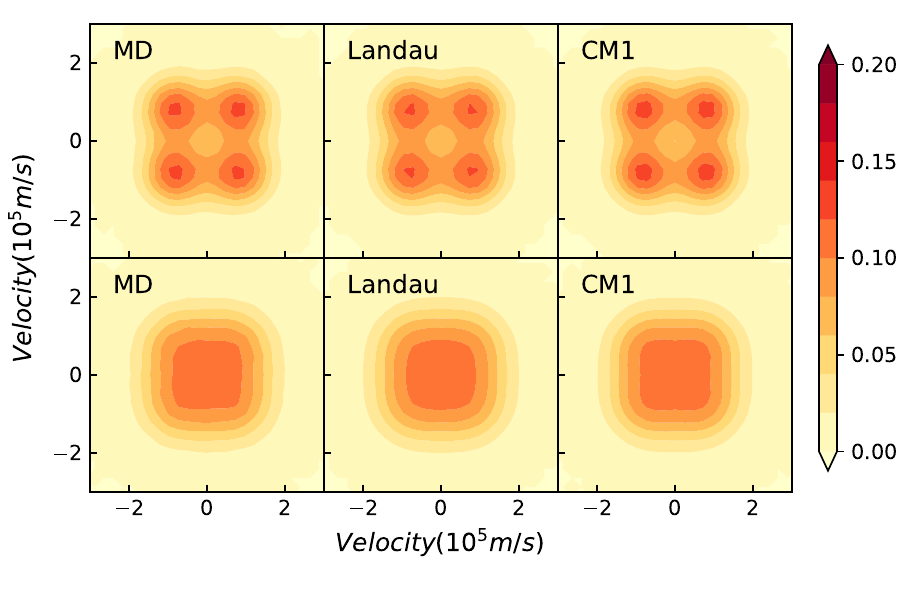}
    \includegraphics[width=0.45\textwidth]{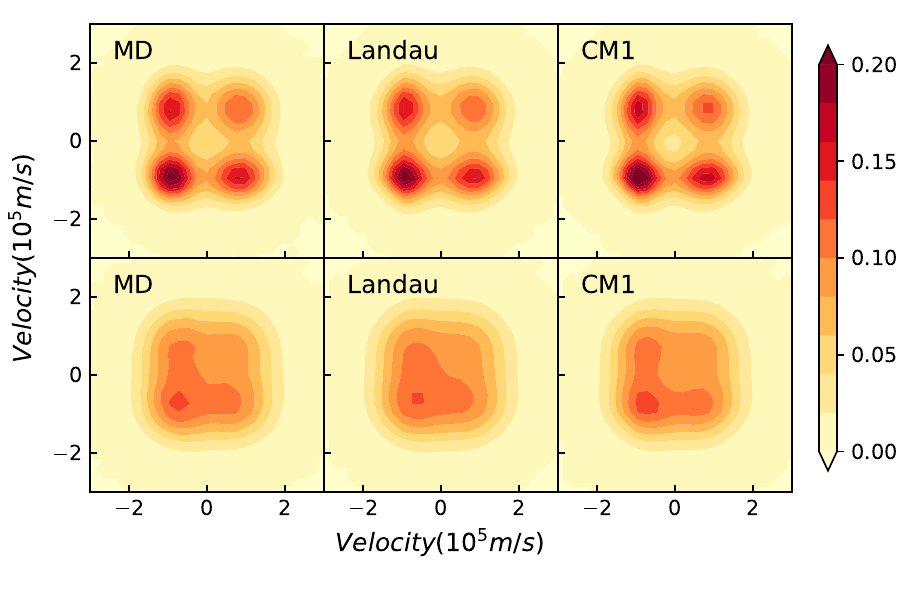}
    \caption{The instantaneous velocity PDF in the $v_{1}\mhyphen v_{2}$ plane from double-well initial distributions predicted by the full MD, the Landau and the CM1 collision model at $t = 1~\text{fs}$ (upper) and $3~\text{fs}$ (lower).}
    \label{fig:100double}
\end{figure}

\begin{figure}[H]
    \centering
    \includegraphics[width=0.45\textwidth]{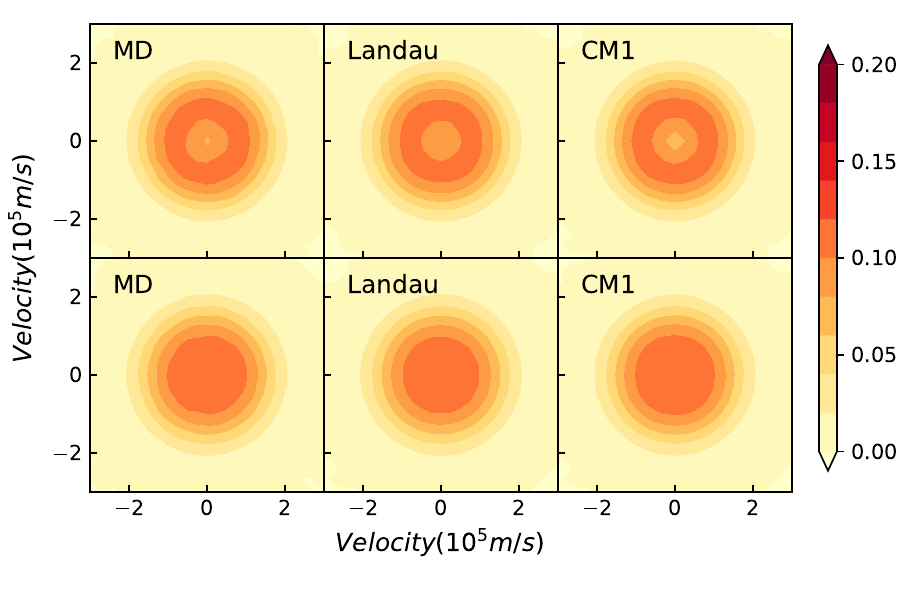}
    \includegraphics[width=0.45\textwidth]{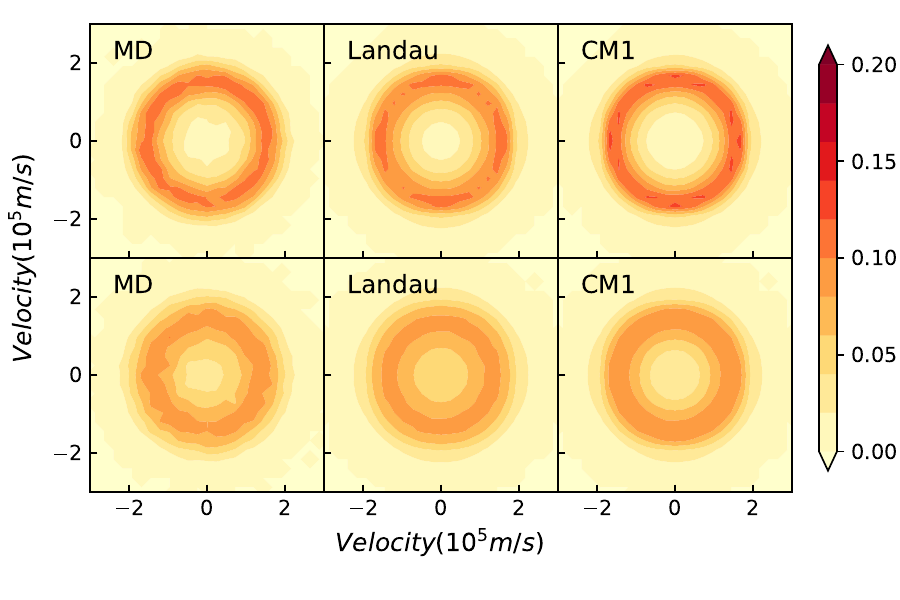}
    \caption{The instantaneous velocity PDF in the $v_{1}\mhyphen v_{2}$ plane from radial oscillation initial distributions predicted by the full MD, the Landau and the CM1 collision model at $t = 1~\text{fs}$ (upper) and $2~\text{fs}$ (lower).}
    \label{fig:100radial}
\end{figure}

\begin{figure}[H]
    \centering
    \includegraphics[width=0.45\textwidth]{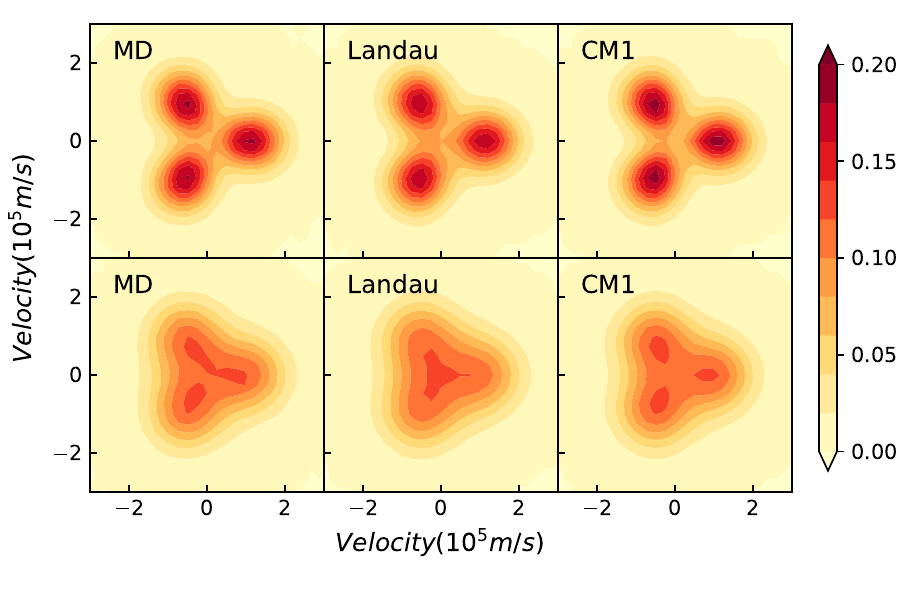}
    \caption{The instantaneous velocity PDF in the $v_{1}\mhyphen v_{2}$ plane from trimodal initial distributions predicted by the full MD, the Landau and the CM1 collision model at $t = 1~\text{fs}$ (upper) and $2~\text{fs}$ (lower).}
    \label{fig:100tri}
\end{figure}

\section{Additional results of the low temperature  regime}\label{sec:LowTemp}

\subsection{Encoder functions of the  CM2 model}

As the temperature decreases to $10 ~ \text{eV}$, the plasma exhibits different kinetic properties.  Unlike the high-temperature regime, the small-angle scattering is no longer dominant; the collective interactions between the pair of collision particles and the environment need to be properly accounted for. In particular, these collective interactions result in a second energy transfer from the collective motion $(\bm v + \bm v')/2$ to the orthogonal plane. Due to the energy conservation, this second energy transfer is restricted to the null space of $\bm u$ (i.e., the projection by $\bm{\mathcal{P}} = \bm{I} - \bm{u}\bm{u}^T/\vert \bm{u}\vert ^{2}$).
As a result, this collective interaction will lead to a net energy transfer from $\widetilde{\bm{r}} \propto \bm{\mathcal{P}}\bm{r}$ to $\widetilde{\bm{s}} \propto \bm{u} \times \bm{r}$, i.e., $g_{r}^2 = g_{u,\perp}^2 - \delta g_{r, \perp}^2$, $g_{s}^2 = g_{u,\perp}^2 + \delta g_{r, \perp}^2$, and therefore $g_r^2 < g_u^2$ in contrast to $g_r^2 \equiv g_u^2 \propto 1/u$ for the Landau model.

The present generalized collision CM2 model enables us to capture this effect, where the encoder functions $g_{r}(u,r,s)$ and $g_{s}(u,r,s)$ denote the energy transfer along the different directions and can be directly learned from the MD results.  Fig. \ref{fig:grgs} shows the 2D contour of the constructed $g_r(u,r,s^{\ast})$ and $g_s(u,r,s^{\ast})$ for $s^{\ast} = 0.2$, $0.4$, and $0.6$, as well as the ensemble average result (the same as Fig. 1 in the main manuscript). We observe that $g_r^2 < g_s^2$ for all of the cases. These numerical results verify
the heterogeneous nature arising from the aforementioned second energy transfer that has been overlooked in the empirical Landau model.

% Next, we examine the evolution of the velocity distribution in a low-temperature high-density plasma at $10 ~ \text{eV}$.
% Numerical simulations are conducted for various initial velocity distributions using the Landau equation, the ``CM1'' and ``CM2'' models, and the results are compared with MD simulations.
% The findings indicate that the Landau equation deviates from the MD results, whereas the ``CM1'' method accurately captures the marginal density distributions.
% The ``CM2'' method further improves numerical accuracy by accounting for interactions between colliding particles and the environment.
% The interaction between the particles and the environment can be expressed in the encoder functions $g_{r}(u,r,s)$ and $g_{s}(u,r,s)$ of the collision operator $\bm{\omega}$.
% Under the Maxwell distribution, the expectation of $s$ is $0.38$ ($\times 10^{10} m^{2}/s^{2}$).
% We compare the two encoder functions with different $s$.
\begin{figure}[H]
    \centering
    \includegraphics[width=0.8\textwidth]{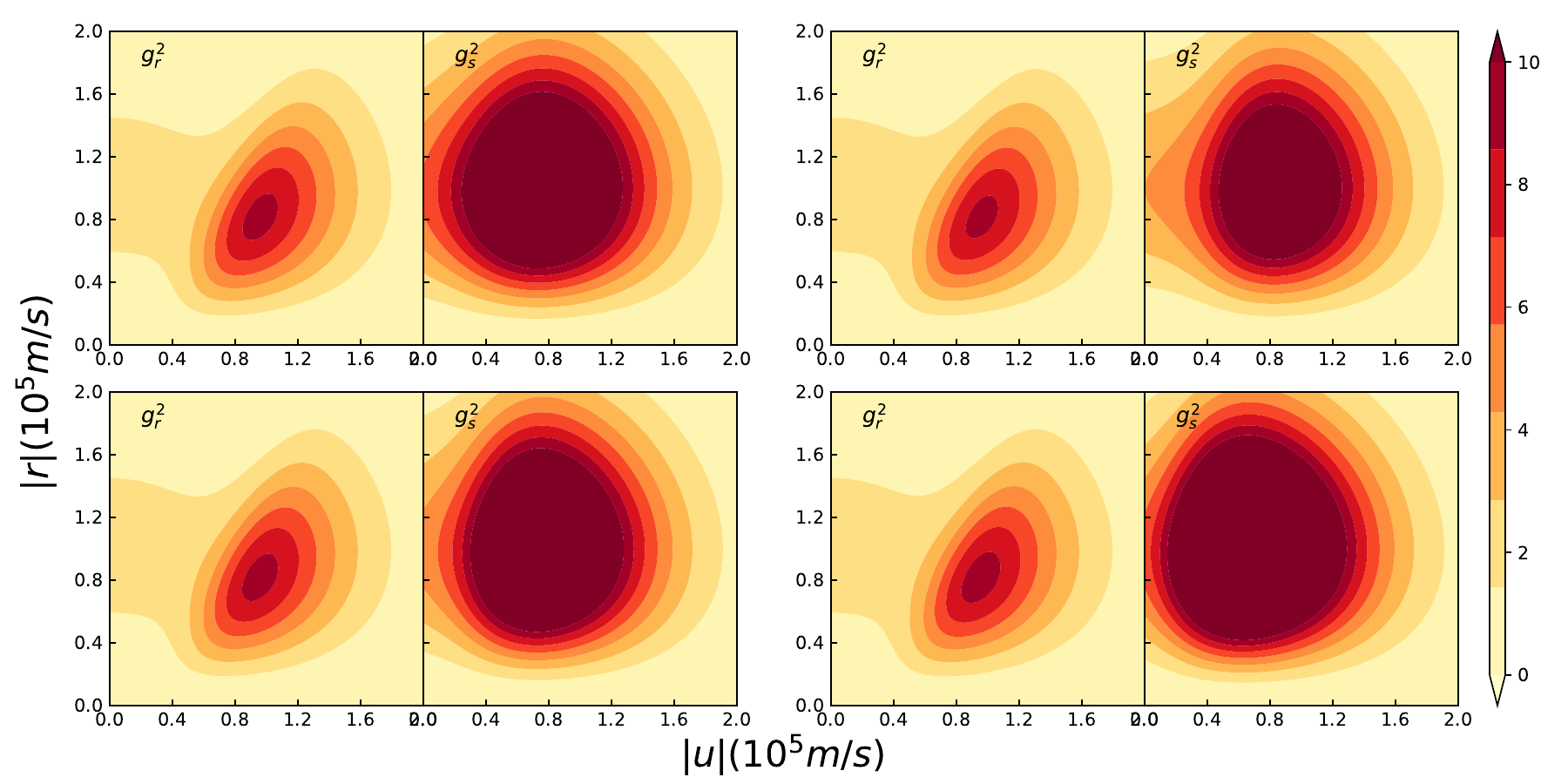}
    \caption{Encode functions $g_{r}(u,r,s)$ and $g_{s}(u,r,s)$ with the ensemble average over $s$ (upper left), $s=0.2$ (upper right), $s=0.4$ (lower left) and $s=0.6$ (lower right). The heterogeneous effect $g_r^2 < g_s^2$ holds for all the cases and verifies the effect of the broadly overlooked second energy transfer arising from the collective interactions between the pair of collision particles and the environment.}
    \label{fig:grgs}
\end{figure}

% Analogous to the Landau equation, the collisions between the colliding particles and the environment lead to an additional energy transfer to the plane perpendicular to $\widetilde{\bm{u}} \propto \bm{v}-\bm{v}'$.
% Our model takes into account further the energy dissipation in the $\widetilde{\bm{r}} \propto \bm{\mathcal{P}}(\bm{v}+\bm{v}')$ direction, which ultimately leads to $g_{r}^{2}<g_{s}^{2}$.
% In this paper, the cross-term $s$ has little influence on the collision operator but may have an effect in more restricted systems.

\subsection{Additional results of the plasma kinetics in the low-temperature regime}

To verify the effectiveness of the present generalized collision model (CM2), we conduct simulations of kinetic processes with the initial conditions following the PDFs in Eqs. \eqref{eq:trainingset} and \eqref{eq:testset}. Figs. \ref{fig:10biMax}, \ref{fig:10double}, \ref{fig:10radial} and \ref{fig:10corr} show the numerical results in comparison with the predictions from the Landau model and the full MD simulations (see Fig. 3 in the main manuscript for the case of the trimodal distribution). For all the cases, the predictions of the present generalized collision model show good agreement with the full MD results. In contrast, the predictions from the Landau model show apparent deviations.  These numerical results reveal the crucial role of the heterogeneous energy transfer effect that has been over-simplified in the Landau model. 

% The simulation results of the ``CM2'' model are performed with the encoder functions $g_{r}$ and $g_{s}$ and compared with the Landau equation and the MD solutions.
% % In addition to the results presented in the main text,
% Additional results are displayed below with various initial velocity distributions.
\begin{figure}[H]
    \centering
    \includegraphics[width=0.45\textwidth]{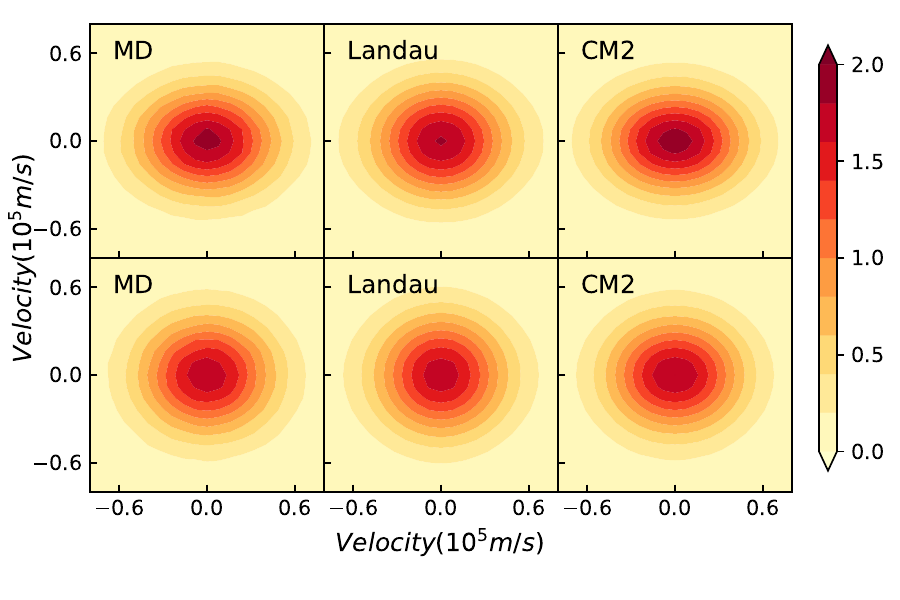}
    \includegraphics[width=0.45\textwidth]{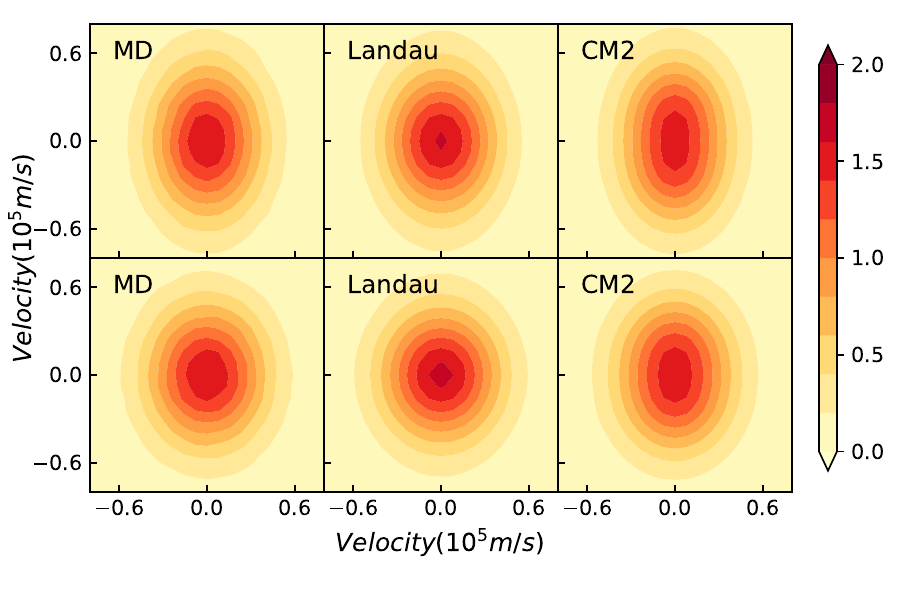}
    \caption{The instantaneous velocity PDF in the $v_{1}\mhyphen v_{2}$ plane from bi-Maxwellian distributions predicted by the full MD, the Landau and the CM2 collision model at $t = 1~\text{fs}$ (upper) and $2~\text{fs}$ (lower).}
    \label{fig:10biMax}
\end{figure}

\begin{figure}[H]
    \centering
    \includegraphics[width=0.45\textwidth]{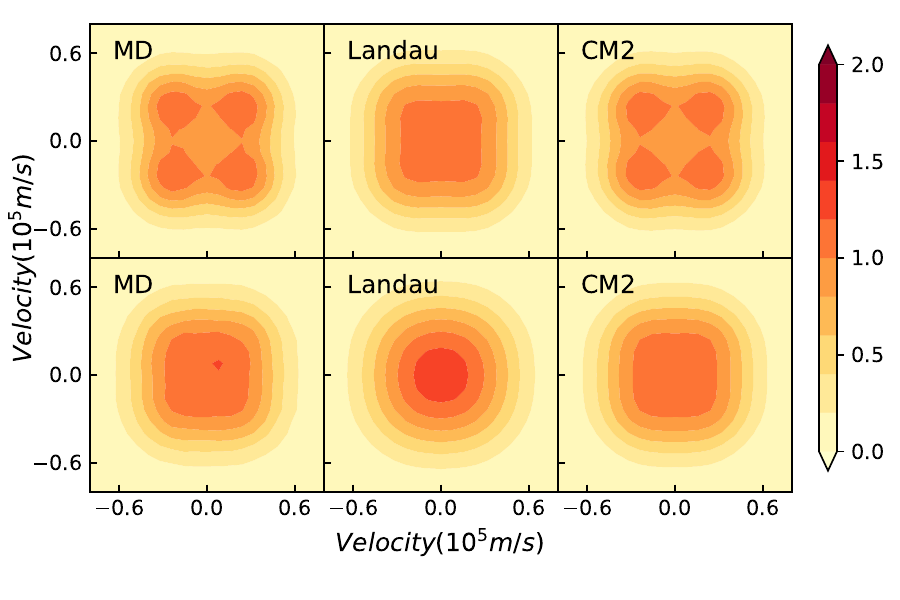}
    \includegraphics[width=0.45\textwidth]{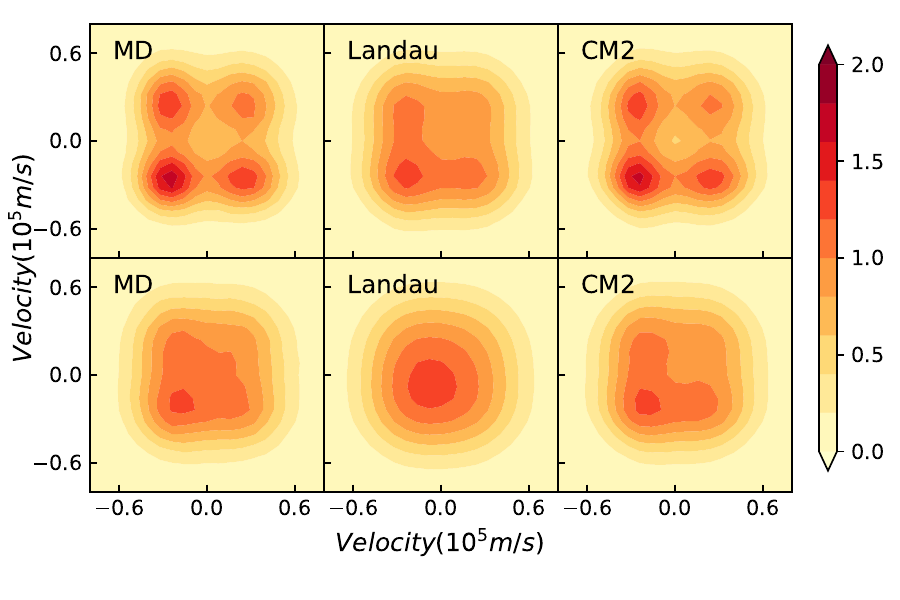}
    \caption{The instantaneous velocity PDF in the $v_{1}\mhyphen v_{2}$ plane from double-well distributions predicted by the full MD, the Landau and the CM2 collision model at $t = 0.4~\text{fs}$ (upper) and $0.8~\text{fs}$ (lower).}
    \label{fig:10double}
\end{figure}

\begin{figure}[H]
    \centering
    \includegraphics[width=0.45\textwidth]{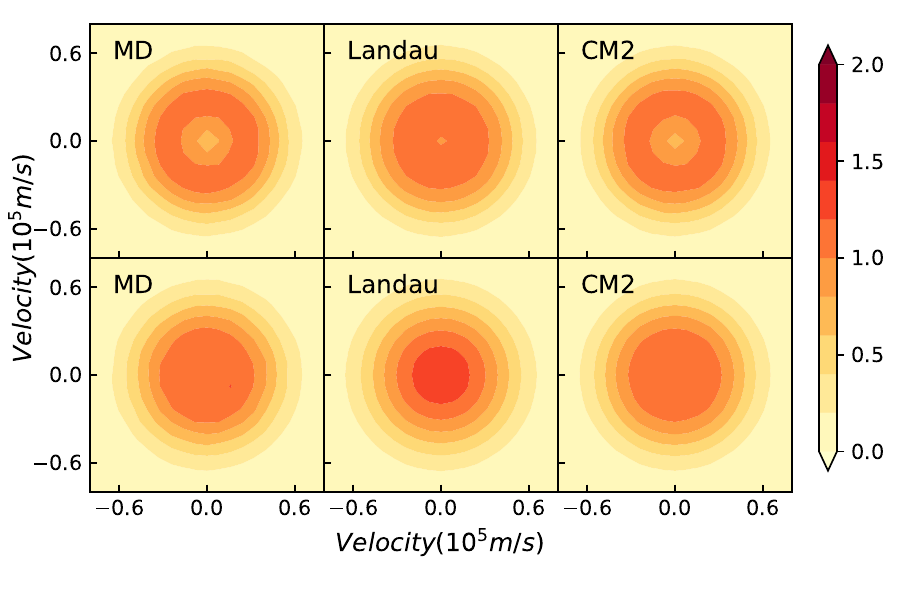}
    \includegraphics[width=0.45\textwidth]{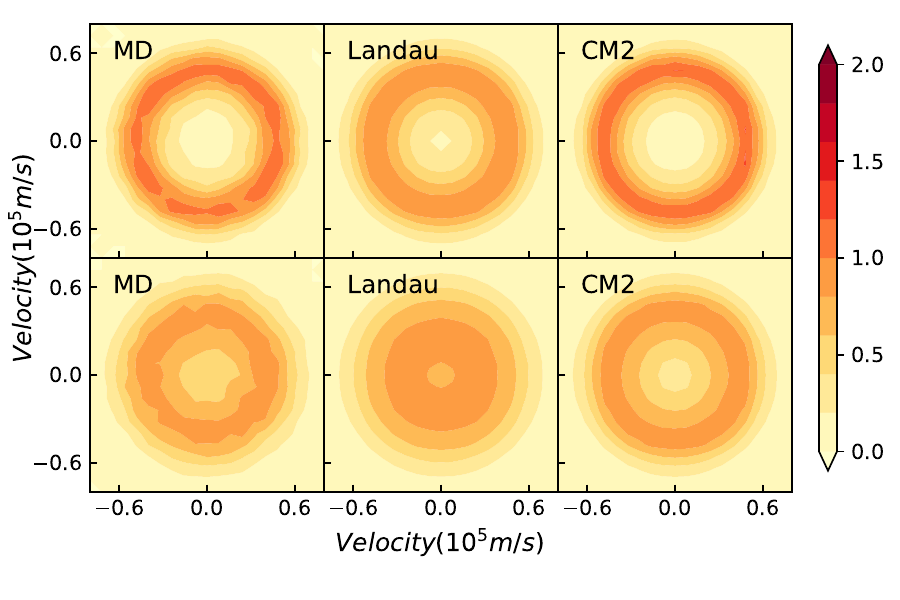}
    \caption{The instantaneous velocity PDF in the $v_{1}\mhyphen v_{2}$ plane from radial oscillation distributions predicted by the full MD, the Landau and the CM2 collision model at $t = 0.2~\text{fs}$ (upper) and $0.6~\text{fs}$ (lower).}
    \label{fig:10radial}
\end{figure}

\begin{figure}[H]
    \centering
    \includegraphics[width=0.45\textwidth]{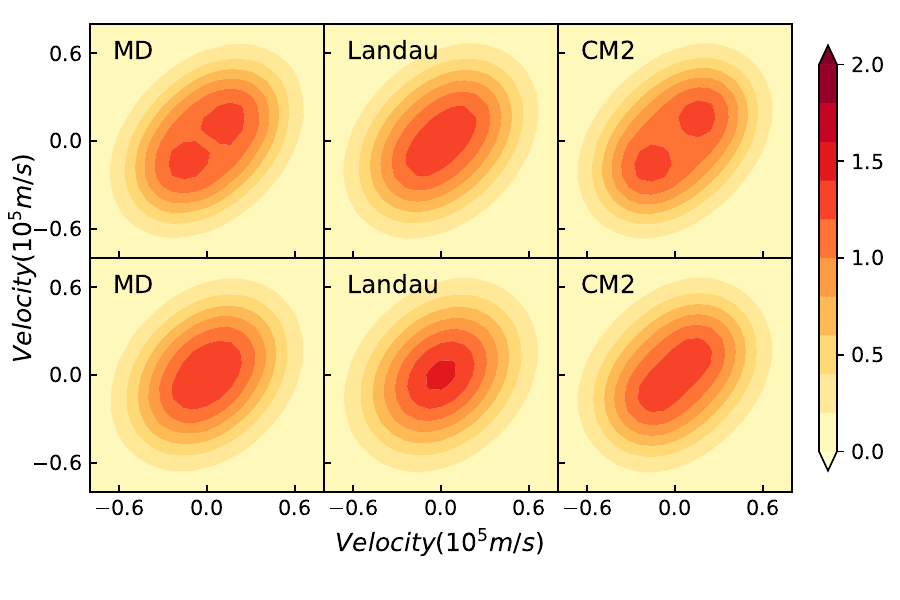}
    \caption{The instantaneous velocity PDF in the $v_{1}\mhyphen v_{2}$ plane from a correlated double-well distribution predicted by the full MD, the Landau and the CM2 collision model at $t = 0.4~\text{fs}$ (upper) and $1~\text{fs}$ (lower).}
    \label{fig:10corr}
\end{figure}

\subsection{Time scaling parameter of the Landau model}

We show that the inaccurate prediction of the Landau model with the kernel $\bm\omega_{\rm Landau} =  c_0/u \mathcal{P}$ shown in Figs. \ref{fig:10biMax}, \ref{fig:10double}, \ref{fig:10radial} and \ref{fig:10corr} are due to the over-simplified formulation, and can not be remedied by re-scaling the time parameter $c_0$. In particular, we choose a new time scale parameter $c_0$ for each case so that the simulation result of the Landau model best matches the MD solutions at $t=0.2~\text{fs}$. However, the Landau model yields inaccurate predictions of the subsequent PDF evolution, as shown in Figs. \ref{fig:time_double}, \ref{fig:time_radial}, \ref{fig:time2_biMax} and \ref{fig:time2_double}. In contrast, the predictions of the ``CM2'' model show good agreement with the full MD results. 

\begin{figure}[H]
    \centering
    \includegraphics[width=0.45\textwidth]{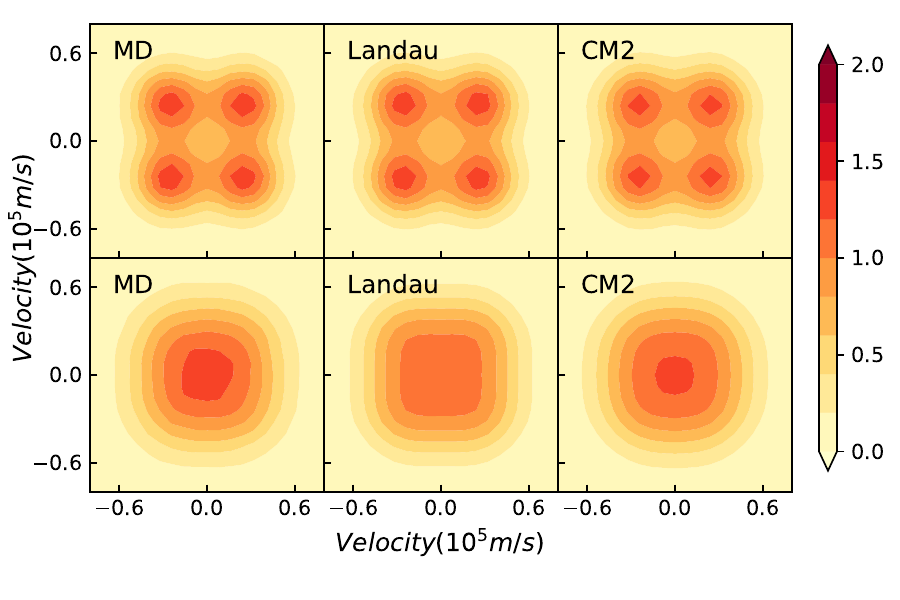}
    \caption{The instantaneous velocity PDF in the $v_{1}\mhyphen v_{2}$ plane from a symmetric double-well distribution predicted by the full MD, the Landau and the CM2 collision model at $t = 0.2~\text{fs}$ (upper) and $1~\text{fs}$ (lower).}
    \label{fig:time_double}
\end{figure}

\begin{figure}[H]
    \centering
    \includegraphics[width=0.45\textwidth]{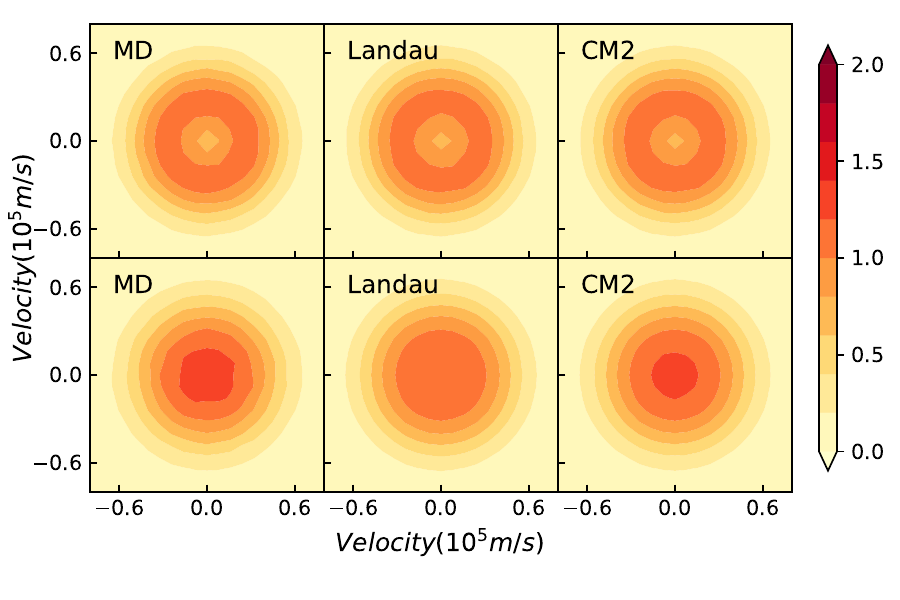}
    \caption{The instantaneous velocity PDF in the $v_{1}\mhyphen v_{2}$ plane from a radial oscillation distribution predicted by the full MD, the Landau and the CM2 collision model at $t = 0.2~\text{fs}$ (upper) and $0.8~\text{fs}$ (lower).}
    \label{fig:time_radial}
\end{figure}

\begin{figure}[H]
    \centering
    \includegraphics[width=0.45\textwidth]{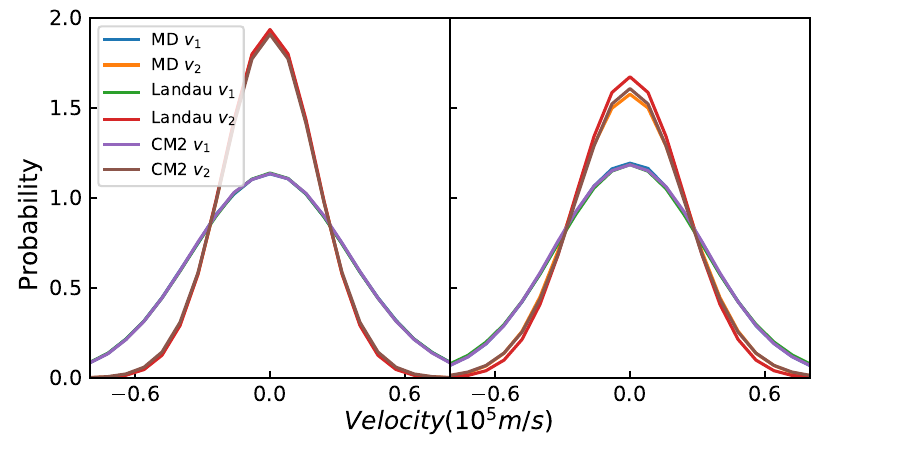}
    \caption{The instantaneous velocity marginal PDF in the $v_{1}$ and $v_{2}$ axis from a bi-Maxwellian distribution predicted by the full MD, the Landau and the CM2 collision model at $t = 0.2~\text{fs}$ (left) and $1~\text{fs}$ (right).}
    \label{fig:time2_biMax}
\end{figure}

\begin{figure}[H]
    \centering
    \includegraphics[width=0.45\textwidth]{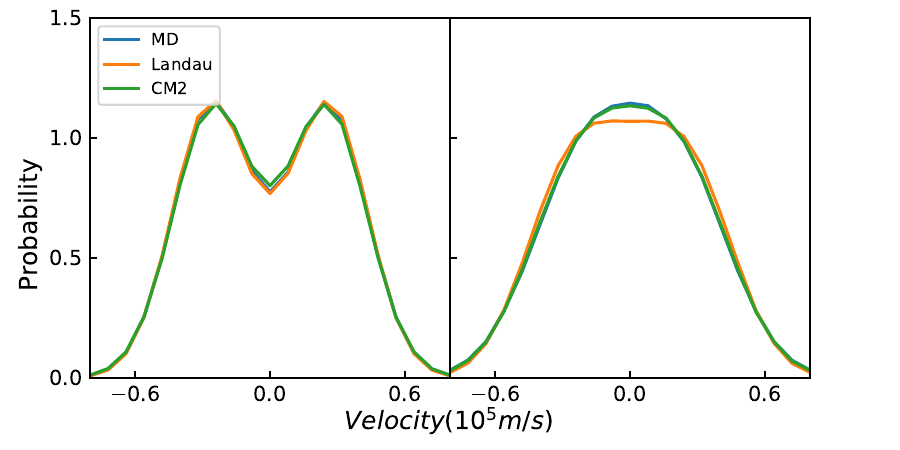}
    \caption{The instantaneous velocity marginal PDF from a double-well distribution predicted by the full MD, the Landau and the CM2 collision model at $t = 0.2~\text{fs}$ (left) and $1~\text{fs}$ (right).}
    \label{fig:time2_double}
\end{figure}

%\bibliographystyle{apsrev4-1}
%\bibliography{ref}
%

\end{document}